\documentclass{article}

\usepackage[utf8]{inputenc}
\usepackage[all,arc]{xy}
\usepackage{enumerate}
\usepackage[top=1in, bottom=1.25in, left=1.25in, right=1.25in]{geometry} 
\usepackage{parskip}
\usepackage{braket}
\usepackage{tikz-cd}
\usepackage{verbatim}
\usepackage{tabu}
\usepackage{ggmaths}

\usepackage{biblatex}
\usepackage{hyperref}
\usepackage{xurl}
\hypersetup{breaklinks=true}
\title{Non-Borda elections under relaxed IIA conditions}
\author{\makeauthorentry}

\begin{document}

\maketitle

\begin{abstract}

Arrow's \autocite{arrow} celebrated Impossibility Theorem asserts that an election rule, or Social Welfare Function (SWF), between three or more candidates meeting a set of strict criteria cannot exist. Maskin \autocite{maskin} suggests that Arrow's conditions for SWFs are too strict. In particular he weakens the ``Independence of Irrelevant Alternatives'' condition (IIA), which states that if in two elections, each voter's binary preference between candidates $c_i$ and $c_j$ is the same, then the two results must agree on their preference between $c_i$ and $c_j$. Instead, he proposes a modified IIA condition (MIIA). Under this condition, the result between $c_i$ and $c_j$ can be affected not just by the order of $c_i$ and $c_j$ in each voter's ranking, but also the number of candidates between them. More candidates between $c_i$ and $c_j$ communicates some information about the strength of a voter's preference between the two candidates, and Maskin argues that it should be admissible evidence in deciding on a final ranking.

We construct SWFs for three-party elections which meet the MIIA criterion along with other sensibility criteria, but are far from being Borda elections (where each voter assigns a score to each candidate linearly according to their ranking). On the other hand, we give cases in which any SWF must be the Borda rule.

\end{abstract}

\section{Introduction}

We let a set of voters be $V$ and the set of candidates $C$, with $|C| = k < \infty$. We always label these candidates as $\set{c_i}_{i=1}^k$, and we treat the indices modulo $k$.

We will consider two cases for the size of $V$. Either it will be finite, in which case we say $|V| = n$, or it will be infinite, in which case we will require that it be equipped with a $\sigma$-algebra $\Sigma$ and a measure $\nu$. SWFs on unmeasurable electorates are not interesting within the context of these investigations due to the need to ``count'' or ``measure'' different sets of voters in order to compare their influences in choosing the final ranking.

In the infinite case, we require that the measure is finite (i.e. $\nu(V) < \infty$), and indeed we normalise to $\nu(V)=1$. Again, this is the critical case because if $\nu(V) = \infty$ then we might have several competing groups all with measure $\infty$, making it impossible to compare them meaningfully. We also require that $\nu$ is atomless (also known as diffuse); that is to say, there is no set $A \in \Sigma$ with $\nu(A)>0$ such that for any $B \subset A$ in $\Sigma$ we have either $\nu(B) = 0$ and $\nu(B) = \nu(A)$. Sierpinski \autocite{sierpinski} showed that this allows us to find subsets of arbitrary measure within any set of larger measure.

We define $\mathfrak{R}$ to be the set of possible votes $r$ that voter $v$ can register. Each $r$ is a total ordering of $C$, and we will treat $r$ as a bijection $r: C \hookrightarrow \set{0, \dots, |C|-1}$, where $r^{-1}(0)$ is the lowest ranked candidate, and $r^{-1}(|C|-1)$ is the highest ranked. We will sometimes write $r$ simply as a list of candidates; for example, $c_1 c_2 c_3$ corresponds to a voter's ballot placing candidate $c_1$ first, followed by $c_2$ and then $c_3$. This corresponds to $r: c_i \mapsto 3-i$.

We define $\bar{E} \subset \mathfrak{R}^V$ to be the set of possible elections. For $|V|$ finite, this is simply $\bar{E} = \mathfrak{R}^V$; if $|V|$ is infinite then $\bar{E} = \set{e: V \rightarrow \mathfrak{R} | e \text{ measurable}}$. Note that $C$ is always finite, so $\mathfrak{R}$ is finite, and we endow it with the discrete $\sigma$-algebra.

Sometimes we will consider SWFs which restrict the set of legal elections. We call the set of legal elections $E \subset \bar{E}$.

An election result $\succ$ is a weak ordering of $C$ - in other words, an equivalence relation of ``tying'', together with a total ordering of ``winning'' between the equivalence classes. We call the set of all possible weak orderings $\mathfrak{R}_=$. In the most general sense, a SWF is any function $F: E \rightarrow \mathfrak{R}_=$. Often $F$ will be determined by some real-valued function $G: C \times E \rightarrow \mathbbm{R}$, where each candidate receives a ``score'' in any election; then $F = \Phi(G)$ is defined by

\begin{equation}
c_i \succ c_j \Leftrightarrow G(c_i,e) > G(c_j,e)
\end{equation}

where $\succ = \Phi(G)(e)$. Since any multiset of reals form a weak ordering under the $>$ relation, this defines a weak ordering for $\succ$.

\subsection{Some definitions}

A voter $v$ is able to submit as their vote some subset of the full set of rankings, $\mathfrak{R}$. We call this the voter's ballot:

\begin{defn}[Ballot]
For a set of legal elections $E$, the ballot $\mathfrak{B}_v$ of a voter $v$ is $\{e(v) | e \in E\}$.
\end{defn}

For the entire paper, we restrict ourselves to a ballot election, in which each voter is given the same ballot, and each voter can choose to submit any ranking on the ballot independently of other voters (up to the constraint of measurability).

\begin{defn}[Ballot domain]
An election has a ballot domain if there exists some $\mathfrak{B} \subset \mathfrak{R}$ with $|\mathfrak{B}| > 1$ such that $E$ is the set of all measurable functions from $V$ to $\mathfrak{B}$.
\end{defn}

The Modified IIA condition means that emphasis is placed on the gap between two candidates in a voter's ballot. For this reason, it helps to define a relative election as follows.

\begin{defn}[Relative election]
Let $D_k = \set{1-k, 2-k, \dots, -1, 1, 2, \dots, k-1}$. This is the set of possible gaps in a ballot between two candidates $c_i$ and $c_j$. Define $\pi_{i, j}(r) = r(c_i)-r(c_j)$; so $\pi_{i,j}: \mathfrak{R} \rightarrow D_k$. Then define $\pi_{i,j}: E \rightarrow D_k^V$ by $\pi_{i,j}(e)(v) = \pi_{i,j}(e(v))$. We call $\pi_{i,j}(e)$ a relative election; it encodes the preference (along with strength of preference) of each voter between candidates $c_i$ and $c_j$.

We let the image of $\pi_{i,j}$ over $E$ be $A_{i,j} \subset D_k^V$. We also define the natural partial order on $D_k^V$, which all $A_{i,j}$ inherit, by $x \geq y$ if and only if $x(v) \geq x(v)$ for all $v$. Moreover, we define the relative ballot $B_{i,j} = \pi_{i,j}(\mathfrak{B})$ for a ballot $\mathfrak{B}$ as $\set{\pi_{i,j}(r) | r \in \mathfrak{B}}$, a subset of $D_k$.

Finally, for an election result $\succ$ we write $\pi_{i,j}(\succ) = W$ if $c_i \succ c_j$, $\pi_{i,j}(\succ)=L$ if $c_j \succ c_i$, and $\pi_{i,j}(\succ) = T$ otherwise. Here $W$ stands for ``win'', $T$ for ``tie'' and $L$ for ``loss''. For convenience we impose the natural total order on these three symbols: $W > T > L$.  Similarly, we will sometimes abuse notation by writing $W = -L$, $L = -W$ and $T=-T$.
\end{defn}

We will often be interested in ordering elections by comparing the performance of one candidate against another in each election. In order to do so, we give a formal definition of the idea of ``promoting'' a candidate with respect to another candidate.

\begin{defn}[Promoting]\label{def:promoting}
Given two ballots $r_1$ and $r_2$, and two candidates $c_i \neq c_j$, if for any third candidate $c_k \neq c_i, c_j$ and any fourth candidate $c_{\ell} \neq c_i, c_j$ in the case of four or more candidates, we have
\begin{enumerate}[(i)]
\item $r_1(c_k) > r_1(c_i) \Rightarrow r_2(c_k) > r_2(c_i)$
\item $r_1(c_j) > r_1(c_k) \Rightarrow r_2(c_j) > r_2(c_k)$
\item $r_1(c_j) > r_1(c_i) \Rightarrow r_2(c_j) > r_2(c_i)$
\item $r_1(c_k) > r_1(c_{\ell}) \Leftrightarrow r_2(c_k) > r_2(c_{\ell})$
\end{enumerate}
then we say that ``$(c_j, c_i)$ is promoted by $(r_2, r_1)$''. For two elections $e_1$ and $e_2$, if for every voter $v$ we have that $(c_j, c_i)$ is promoted by $(e_2(v),e_1(v))$ then we further say that $(c_j, c_i)$ is promoted by $(e_2, e_1)$.
\end{defn}

\begin{rem}
It follows immediately from the definition of a promotion that if $(c_j,c_i)$ is promoted by $(r_3,r_2)$ and by $(r_2,r_1)$ then it is also promoted by $(r_3,r_1)$; and the same goes for elections $e_1, e_2, e_3$.
\end{rem}

Now we are ready to define a family of ballots for which our intuitions about ``better'' and ``worse'' results correspond directly to the partial order on relative elections given above.

\begin{defn}[Increasing ballot]
A ballot $\mathfrak{B}$ is increasing if, for any candidates $c_i$ and $c_j$, and for any $m, m' \in \pi_{i,j}(\mathfrak{B})$ with $m>m'$, and any $r \in \mathfrak{B}$ with $\pi_{i,j}(r) = m$, there exists $r' \in \mathfrak{B}$ with $\pi_{i,j}(r') = m'$ and $(c_j, c_i)$ is promoted by $(r', r)$.
\end{defn}

\begin{defn}[Intermediate ballot]
A ballot $\mathfrak{B}$ is intermediate if, for any candidates $c_i$ and $c_j$ and for any $r_-, r_+ \in \mathfrak{B}$ with $(c_j,c_i)$ promoted by $(r_-, r_+)$, and for any $m \in \pi_{i,j}(\mathfrak{B})$ with $\pi_{i,j}(r_-) < m < \pi_{i,j}(r_+)$, there exists $r \in \mathfrak{B}$ with $\pi_{i,j}(r) = m$ such that $(c_j,c_i)$ is promoted both by $(r_-,r)$ and by $(r,r_+)$.
\end{defn}

Later we will present lemmas \ref{amiia}, \ref{prmiia} and \ref{ivmiia} confirming that elections on increasing, intermediate ballot domains behave as one would expect based on their quantitative properties.

\subsection{Summary of results}

In section \ref{section:conditions} we will define various conditions on SWFs. We define the unrestricted domain, in which no election is ruled out by restricting $E$, and the Condorcet cycle domain, which is the ballot domain with three candidates and a fixed cyclic order of preference between them.

We then define the critical conditions required by Maskin in \cite{maskin}, first and foremost the MIIA condition, which requires the relative ranking between two candidates to be fixed by the gap between them in each voter's ranking. We also define the anonymity condition A, which requires all voters to be treated anonymously, and the neutrality condition N, which requires that the candidates can be permuted.

As long as the election domain is sufficiently well-behaved, we prove in Lemmas \ref{amiia}, \ref{nmiia} and \ref{namiia} that these conditions allow us to simplify our SWF $F$ to a single weight-based relative SWF $f'$, which gives the final ranking between any two candidates, given only the sizes of the parts of the partition of $V$ according to the gap between those two candidates in each voter's ballot.  In Lemmas \ref{intermediateballots}, \ref{increasingballots} and \ref{separableballots} we confirm that the unrestricted and Condorcet cycle domains are well-behaved in all the ways required.

We then define various conditions which guarantee that the SWF behaves as expected when a candidate does ``well'' or ``poorly''. The Pareto principle states that a unanimous preference among voters between two candidates forces that preference in the final ranking; Positive Responsiveness ensures that moving a candidate up in the voters' rankings cannot harm that candidate in the final ranking; and the Intermediate Value condition states for $V$ infinite that this can be done smoothly. In Lemmas \ref{prmiia} and \ref{ivmiia} we demonstrate that under the MIIA condition, PR and IV are equivalent to simpler statements about the relative SWFs that the MIIA condition gives.

We also give in Lemma \ref{relativetofull} sufficient criteria for a relative SWF, as defined in the MIIA condition, to correspond to a valid SWF $F$. This expands on the method used by Kalai in \autocite{kalai}. This allows us to construct an SWF simply by giving a relative SWF and checking these criteria. We complete our preparation for the main results by defining the Borda rules in section \ref{section:borda}. These are the rules characterised by each voter awarding scores to each candidate, linear to the candidate's placement in the voter's ballot, after which the candidates are compared by their aggregate score across all voters. In Lemma \ref{weightinglemma} we demonstrate that the anonymity condition on the SWF guarantees that any Borda rule is equivalent to an unweighted Borda rule (i.e. where each voter awards scores on the same scale).

We then construct non-Borda SWFs meeting various combinations of these conditions as follows.

The most important construction in the context of social choice theory is one which meets all of the strictest conditions - unrestricted domain, anonymity, neutrality, MIIA, positive responsiveness, Pareto and intermediate value. In section \ref{UA} we give these constructions for $|C|=3$ both in the infinite case (in Theorem \ref{IUA}) and in the finite case for large enough $n$ (in Theorem \ref{FUA}). The case of $|C|>3$ is addressed in a forthcoming paper.

These constructions become Borda when restricted to the Condorcet cycle domain, which raises the question of whether non-Borda rules can exist under condition CC. We address this in section \ref{CA}.

The finite case is answered by Theorem \ref{FCA}, which states that even without any increasingness condition, P or PR, the SWF must be Borda. However, in the infinite case, Theorem \ref{ICAchoice} gives more non-Borda rules even in the Condorcet cycle domain, relying on the axiom of choice (with the Pareto principle but without the PR condition). But we eliminate these rules if we require that $F$ is measurable, using the natural $\sigma$-algebra on $E$ given in Definition \ref{def:measurable}.

Imposing measurability, we can prove Theorem \ref{ICAmeasurable}, in which we do not need either PR or P to guarantee that the SWF is Borda. Nonetheless, PR implies measurability, allowing us to give as Corollary \ref{ICA} a new proof of the result appearing in \autocite{maskin} that on the Condorcet cycle domain, MIIA, A, N and PR imply that $F$ is the positive unweighted Borda rule or the tie rule\footnote{In \autocite{maskin} the tie rule is precluded by a slightly stronger definition of the PR condition in which any ``promotion'' of $c_j$ over $c_i$ when previously they had tied must result in $c_j$ beating $c_i$ - in a sense this describes the function as being ``strictly increasing'' rather than just ``increasing'').}.

\section{Conditions on SWFs}\label{section:conditions}

We now define a series of conditions on SWFs.

\begin{defn}[Unrestricted Domain - U]
A SWF has unrestricted domain if it determines social preferences for as large as possible a family of measurable elections. In other words, $E = \bar{E}$. This is equivalent to taking the ballot domain with $\mathfrak{B} = \mathfrak{R}$.
\end{defn}

\begin{defn}[Condorcet Cycle Domain - CC]
For $|C|=3$, a SWF has Condorcet cycle domain if it determines social preferences for elections in which all votes are even permutations of $c_1 c_2 c_3$. Writing $\mathfrak{C} = \set{c_1 c_2 c_3, c_2 c_3 c_1, c_3 c_1 c_2}$, the Condorcet Cycle condition means that $E = \bar{E} \cap \mathfrak{C}^V$. In other words, this is the ballot domain with $\mathfrak{B} = \mathfrak{C}$.
\end{defn}

Note that all SWFs which satisfy U can be restricted to SWFs which satisfy CC. However, we will see examples where SWFs on the unrestricted domain can behave unexpectedly but, after restricting them to the Condorcet cycle domain, their behaviour becomes more predictable. Therefore it is of independent interest to try to construct unexpected SWFs on the Condorcet cycle domain. (It is also feasible that a SWF satisfying U and IV would lose the IV property when restricted to the Condorcet cycle domain, but we do not cover any such cases here).

The following results mean that as long as we impose one of U and CC, we can always assume the numeric properties of our domain that come with an increasing, intermediate domain.

\begin{lem}\label{intermediateballots}
$\mathfrak{R}$ and $\mathfrak{C}$ are intermediate ballots.
\end{lem}

\begin{proof}
For $\mathfrak{C}$ and for any $i \neq j$, either $\pi_{i,j}(\mathfrak{C}) = \set{1,-2}$ or $\pi_{i,j}(\mathfrak{C}) = \set{2,-1}$; in either case, there can be no $\pi_{i,j}(r_-) < m < \pi_{i,j}(r_+)$ in the set, so intermediacy holds vacuously.

On the other hand, for $\mathfrak{R}$, suppose we have $\pi_{i,j}(r_-) < m < \pi_{i,j}(r_+)$ with $(c_j,c_i)$ promoted by $(r_+, r_-)$. We let $\pi_{i,j}(r_+)=m_+$ and define $r_{m_+} = r_+$. Now given $r_1$ we define $r_{-1}$ by transposing $c_i$ and $c_j$; that is, $r_{-1}(c_i) = r_1(c_j)$ and $r_{-1}(c_j) = r_1(c_i)$. On the other hand, given $r_{m'}$ for some $m_- < m' \leq m_+$ with $m' \neq 1$, we have $r_{m'}(c_i) - r_{m'}(c_j) > r_-(c_i)-r_-(c_j)$, so either $r_{m'}(c_i)> r_-(c_i)$ or $r_{m'}(c_j) < r_-(c_j)$. In the first case, we let $c_k \neq c_j$ be such that $r_{m'}(c_k) = r_{m'}(c_i)-1$, and we transpose $c_i$ and $c_k$ to get $r_{m'-1}$ with $r_{m'-1}(c_i) = r_{m'}(c_k)$ and $r_{m'-1}(c_k)=r_{m'}(c_i)$. Symmetrically, in the second case, we let $c_k \neq c_i$ be such that $r_{m'}(c_k) = r_{m'}(c_j)+1$, and we transpose $c_j$ and $c_k$ to get $r_{m'-1}$ with $r_{m'-1}(c_j) = r_{m'}(c_k)$ and $r_{m'-1}(c_k)=r_{m'}(c_j)$.

Clearly we have $\pi_{i,j}(r_{m'-1}) = \pi_{i,j}(r_{m'})-1$, and as long as $\pi_{i,j}(r_1)=1$ we get $\pi_{i,j}(r_{-1})=-1$. Moreover $\pi_{i,j}(r_{m_+}) = m_+$, so inductively we have $\pi_{i,j}(r_{m'}) = m'$ for $m_- \leq m' \leq m_+$. Moreover, each transposition guarantees that $(c_j,c_i)$ is promoted by $(r_{m'-1},r_{m'})$, and also by $(r_{-1},r_1)$. Finally $r_{m_-}(c_i) \geq r_-(c_i)$ and $r_{m_-}(c_j) \leq r_-(c_j)$, but $r_{m_-}(c_i) - r_{m_-}(c_j) = r_-(c_i)-r_-(c_j)$ so $r_{m_-}(c_i) = r_-(c_i)$ and $r_{m_-}(c_j) = r_-(c_j)$. Now for any two other candidates $c_k$ and $c_{\ell}$, we have $r_-(c_k) > r_-(c_{\ell})$ if and only if $r_+(c_k) > r_+(c_{\ell})$, since $(c_j,c_i)$ is promoted by $(r_+,r_-)$; but $c_k$ and $c_{\ell}$ have never been transposed to construct any intermediate $r_{m'}$, so $r_+(c_k) > r_+(c_{\ell})$ if and only if $r_{m_-}(c_k) > r_{m_-}(c_{\ell})$.

Therefore $r_-$ and $r_{m_-}$ are identical. So we have broken the promotion $(r_-,r_+)$ into a series of promotions covering each intermediate value between $m_-$ and $m_+$, so indeed there exists $r_m$ such that $\pi_{i,j}(r_m)=m$ and $(c_j,c_i)$ is promoted both by $(r_-, r_m)$ and by $(r_m, r_+)$ (we can chain together promotions, as remarked after Definition \ref{def:promoting}).
\end{proof}

\begin{lem}\label{increasingballots}
$\mathfrak{R}$ and $\mathfrak{C}$ are increasing ballots.
\end{lem}

\begin{proof}

In $\mathfrak{R}$, we have $\pi_{i,j}(\mathfrak{R})=D_k$. Take any $m, m' \in \pi_{i,j}(\mathfrak{B})$ with $m>m'$, and any $r \in \mathfrak{B}$ with $\pi_{i,j}(r) = m$. Now we define $r_-$ to be the unique ranking for which all $c_k \neq c_i, c_j$ are in the same order as in $r$, but $c_i$ is sent to the lowest position (so $r_-(c_i)=0$) and $c_j$ to the highest position (so $r_-(c_j) = k-1$). It is immediate that $(c_j,c_i)$ is promoted by $(r_-,r)$. Now if $m' = 1-k$ then setting $r' = r_-$ gives that the ballot $\mathfrak{R}$ is increasing. Otherwise, we have $m > m' > 1-k$, and we apply the intermediate property of $\mathfrak{R}$ proved in Lemma \ref{intermediateballots} to find $r'$ with $\pi_{i,j}(r')=m'$ and $(r',r)$ promoted by $(c_j,c_i)$.

In $\mathfrak{C}$, we must separately deal with the cases of $j = i+1$ and $j = i-1$. In the former case, we have $\pi_{i, i+1}(\mathfrak{C}) = \set{1, -2}$, and we must have $m=1$ and $m' = -2$. We always set $r' = c_{i+1} c_{i+2} c_i$ and we notice that whether $r = c_i c_{i+1} c_{i+2}$ or $r = c_{i+2},c_i,c_{i+1}$, $(c_{i+1},c_i)$ is promoted by $(r',r)$. Similarly, in the latter case, we have  $\pi_{i,i+2}(\mathfrak{C}) = \set{2,-1}$; now $r$ must be $c_i c_{i+1} c_{i+2}$, and $r'=c_{i+1} c_{i+2} c_i$ has that $(c_{i+2},c_i)$ is promoted by $(r',r)$ as required.
\end{proof}

Thus, either U or CC is sufficient to guarantee that the election domain is an increasing, intermediate ballot domain.

The following condition was introduced by Maskin in \autocite{maskin}, as an alternative to Arrow's stricter IIA condition.

\begin{defn}[Modified Irrelevance of Independent Alternatives - MIIA]
A SWF $F$ satisfies the MIIA condition if, for two candidates $c_i$ and $c_j$, there exists a ``relative SWF'': a function $f_{i,j}: A_{i,j} \rightarrow \set{W, T, L}$ such that for all $e \in E$ we have $f_{i,j}(\pi_{i,j}(e)) = \pi_{i,j}(F(e))$.
\end{defn}

Intuitively, if for two elections, each voter places the same order between candidates $c_1$ and $c_2$ with the same number of other candidates between them, then the results of those two elections must give the same ranking between $c_1$ and $c_2$.

Often this function $f_{i,j}$ will be given by the sign of some real value, so we define $\varphi: \mathbbm{R} \rightarrow \set{W,T,L}$ by $\varphi(x) = W$ for $x>0$, $\varphi(x) = T$ for $x=0$ and $\varphi(x) = L$ for $x<0$. Note that if $F = \Phi(G)$ for some $G$, we have $\varphi(G(c_i,e) - G(c_j,e)) = f_{i,j}(\pi_{i,j}(e))$.

We now define the anonymity condition, which states that the SWF is unaffected by which voters submit which ballots, and instead is determined only by \emph{how many} voters (either by count or by measure) vote in a given way. In light of this, we define $\tau: E \rightarrow \mathbbm{R}_{\geq 0}^{\mathfrak{R}}$ by $\tau(e)(r) = \nu(e^{-1}(r))$ (or for $V$ finite $\tau(e)(r) = |e^{-1}(r)|$), and we let $E' = \Image(\tau)$.

\begin{defn}[Anonymity - A]
A SWF $F$ satisfies anonymity if there exists a function $F': E' \rightarrow \mathfrak{R}_=$ such that $F = F' \circ \tau$.
\end{defn}

\begin{rem}
$F$ and $F'$ determine one another. If $F'_1 = F'_2$ then $F_1 = F'_1 \circ \tau = F'_2 \circ \tau = F_2$, and if $F'_1 \neq F'_2$ then we can pick $\vec{\varepsilon} \in E'$ on which $F'_1$ and $F'_2$ disagree; then we pick $e \in E$ with $\tau(e)=\vec{\varepsilon}$, and $F_1(e) = F'_1(\vec{\varepsilon}) \neq F'_2(\vec{\varepsilon}) = F_2(e)$, so $F_1 \neq F_2$.
\end{rem}

We can think of $F'$ as the ``weight-based'' SWF. For a finite SWF satisfying anonymity, this is equivalent to the condition that the voters can be permuted in any way.

\begin{lem}\label{permutingvoters}
Let $F$ be a SWF on $V$ finite. Then $F$ satisfies anonymity if and only if for any permutation $\sigma \in \Sym(V)$ we have $F(e \circ \sigma) = F(e)$.
\end{lem}

\begin{proof}
In the forward direction, suppose $F$ satisfies anonymity. Then $F = F' \circ \tau$. So it suffices to show that $\tau(e \circ \sigma) = \tau(e)$. Indeed, 

\begin{equation}
\begin{aligned}
\tau(e \circ \sigma)(r) &= |(e \circ \sigma)^{-1}(r)| \\
&= |\sigma^{-1}(e^{-1}(r))| \\
&= |e^{-1}(r)| \\
&= \tau(e)(r)
\end{aligned}
\end{equation}

The penultimate equality holds because $\sigma$ is a bijection, so $\sigma^{-1}(X)$ and $X$ have the same cardinality for any $X$.

For the converse, suppose that two elections $e$ and $e'$ have $\tau(e) = \tau(e')$; it will suffice to show that $F(e) = F(e')$. For each $r \in \mathfrak{R}$, note that $\tau(e)(r) = \tau(e')(r)$, so $|e^{-1}(r)| = |e'^{-1}(r)|$. Therefore we can define a bijection $m_r: e^{-1}(r) \rightarrow e'^{-1}(r)$. Now we let $\sigma$ be the union of these $m_r$ over each $r$, which is a permutation of $V$. Then $e = e' \circ \sigma$, so $F(e') = F(e' \circ \sigma) = F(e)$, as required.
\end{proof}

Note that if we also have the MIIA condition, then the functions $f_{i,j}$ can also be replaced by $f'_{i,j}$ functions in a corresponding manner. We let $\tau_{i,j}: A_{i,j} \rightarrow \mathbbm{R}_{\geq 0}^{D_k}$ by $\tau_{i,j}(a)(i) = \nu(a^{-1}(i))$, or $|a^{-1}(i)|$ for $V$ finite, and let $A'_{i,j} = \Image(\tau_{i,j})$. We also define $\varrho: A'_{i,j} \rightarrow A'_{j,i}$ by $\varrho(\vec{\alpha})_k = \alpha_{-k}$.

\begin{rem}
We can define $\pi_{i,j}$ to act on $E'$ giving a result in $A'_{i,j}$, by setting
\begin{equation}
\pi_{i,j}(\vec{\varepsilon})_m = \sum_{r: \pi_{i,j}(r)=m} \varepsilon_m
\end{equation}
Then $\pi_{i,j} \circ \tau = \tau_{i,j}\circ \pi_{i,j}$.
\end{rem}

Now there exists $f'_{i,j}$ such that $f_{i,j} = f'_{i,j} \circ \tau_{i,j}$, by the following lemma:

\begin{lem}\label{amiia}
Let $F$ be a SWF on a ballot domain satisfying MIIA and A. Then for any two candidates $c_i$ and $c_j$ we can define a weight-based relative SWF $f'_{i,j}$ such that $f_{i,j} = f'_{i,j} \circ \tau_{i,j}$.
\end{lem}

Indeed, $f_{i,j}$ and $f'_{i,j}$ will be in one to one correspondence; if $f'_{i,j} = g'_{i,j}$ then $f_{i,j} = f'_{i,j} \circ \tau_{i,j} = g'_{i,j} \circ \tau_{i,j} = g_{i,j}$. On the other hand, if $f'_{i,j} \neq g'_{i,j}$ then we can pick $\vec{\alpha} \in A'_{i,j}$ on which they disagree; then we can choose $e \in E$ such that $\tau_{i,j}(\pi_{i,j}(e)) = \vec{\alpha}$. But then

\begin{equation}
\begin{aligned}
f_{i,j}(\pi_{i,j}(e)) &= f'_{i,j}(\vec{\alpha}) \\
&\neq g'_{i,j}(\vec{\alpha}) \\
&= g_{i,j}(\pi_{i,j}(e))
\end{aligned}
\end{equation}

so $f_{i,j}$ and $g_{i,j}$ disagree on some $a = \pi_{i,j}(e)$ in $A_{i,j}$.

\begin{proof}
For any $\vec{\alpha} \in A'_{i,j}$, we need to prove that all $a \in \tau_{i,j}^{-1}(\vec{\alpha})$ take the same value $f_{i,j}(a)$. Then we can define $f'_{i,j}(\vec{\alpha}) = f_{i,j}(a)$ for any choice of $a$.

Suppose that $a_1$ and $a_2$ in $A_{i,j}$ have $\tau_{i,j}(a_1) = \tau_{i,j}(a_2) = \vec{\alpha}$. For each $s \in D_k$ with $\vec{\alpha}_s > 0$, let $r_s \in \mathfrak{D}$ with $\pi_{i,j}(r_s)=s$. Note that the function $s \mapsto r_s$ is an injection. Now, since $F$ is on a ballot domain, we can define $e_1$ and $e_2$ by $e_1(v) = r_{a_1(v)}$ and $e_2(v) = r_{a_2(v)}$; so $\pi_{i,j}(e_i)=a_i$ for $i=1,2$. Then $\nu(e_1^{-1}(r_s)) = \nu(a_1^{-1}(s)) = \vec{\alpha}_s$, and the same goes for $\nu(e_2^{-1}(r_s))$. So for all $r \in \mathfrak{R}$, either $r$ is not equal to $r_s$ for any $s$, and $\tau(e_1)(r) = \tau(e_2)(r)=0$, or $r=r_s$ for some $s$, and $\tau(e_1)(r) = \tau(e_2)(r) = \vec{\alpha}_s$. Hence $\tau(e_1) = \tau(e_2)$, so $F(e_1) = F(e_2)$ by condition A.

Now $f_{i,j}(a_1) = \pi_{i,j}(F(e_1)) = \pi_{i,j}(F(e_2)) = f_{i,j}(a_2)$ as required.
\end{proof}

We define a partial order on $A'_{i,j}$ corresponding to the order we already have on $A_{i,j}$. We say that $\vec{\alpha} \geq \vec{\beta}$ if, for all $m \in D_k$, we have

\begin{equation}
\sum_{n\geq m} \alpha_n \geq \sum_{n\geq m} \beta_n
\end{equation}

In other words, $\vec{\alpha}$ majorises $\vec{\beta}$. Clearly if $a \geq b \in A_{i,j}$ then $\tau_{i,j}(a) \geq \tau_{i,j}(b)$.

We can define a whole family of weaker anonymity conditions by requiring $F$ to factor into some $F'' \circ \tau'$ for $\tau'$ a refinement of $\tau$ (i.e. with $\tau = \chi \circ \tau'$ for some $\chi$). Similarly, for $V$ finite, we can define a weaker anonymity conditions by requiring $F$ to be invariant only under some subgroup of the permutations $\Sym(V)$, as opposed to the entire set. The important case of transitive anonymity, in which we are given only that this subgroup $G \leq \Sym(V)$ acts transitively on $V$, will be addressed in forthcoming papers.

As well as some amount of symmetry between voters we expect symmetry between candidates. For a permutation $\rho: C \rightarrow C$ we can apply $\rho$ to $r \in \mathfrak{R}$ or to $\succ \in \mathfrak{R}_=$ in the natural way: we define $\rho(r) = r \circ \rho$ and $\rho(\succ)$ to be the unique weak ordering such that for all $c_i \neq c_j \in C$ with $\rho(c_i)=c_{i'}$ and $\rho(c_j)=c_{j'}$ we have $\pi_{i',j'}(\rho(\succ)) = \pi_{i,j}(\succ)$.

\begin{defn}[Neutrality - N]
A SWF satisfies neutrality if the candidates are treated symmetrically. That is to say, for any permutation $\rho: C \rightarrow C$ and for all $e$ such that $\rho \circ e \in E$ we have $F(\rho \circ e) = \rho(F(e))$.
\end{defn}

We would hope that an election satisfying both N and MIIA has $f_{i,j}$ and $f_{i',j'}$ agreeing on their intersections for any two pairs $(c_i,c_j)$ and $(c_{i'},c_{j'})$. To guarantee this, we require another property of the domain:

\begin{defn}[Separable domain]
An election domain $E$ is separable if, for any two pairs of candidates $c_i \neq c_j$ and $c_{i'} \neq c_{j'}$, if the intersection between $A_{i,j}$ and $A_{i',j'}$ is non-empty then there exists a permutation $\rho \in \Sym(C)$ such that $\rho(E) = E$, $\rho(c_i)=c_{i'}$ and $\rho(c_j)=c_{j'}$.
\end{defn}

\begin{lem}\label{nmiia}
Let $F$ be a SWF on a separable domain satisfying MIIA and N. Then for any candidates $c_i \neq c_j$ and $c_{i'} \neq c_{j'}$, if $A_{i,j} \cap A_{i',j'}$ is non-empty then $A_{i,j}=A_{i',j'}$ and $f_{i,j}=f_{i',j'}$.
\end{lem}

\begin{proof}
Since $A_{i,j} \cap A_{i',j'}$ is non-empty, the definition of separable domain gives the existence of $\rho \in \Sym(C)$ such that $\rho(E) = E$, $\rho(c_i)=c_{i'}$ and $\rho(c_j)=c_{j'}$. Moreover

\begin{equation}
\begin{aligned}
A_{i,j} &= \pi_{i,j}(E) \\
&= \pi_{i',j'}(\rho(E)) \\
&= \pi_{i',j'}(E) \\
&= A_{i',j'}
\end{aligned}
\end{equation}

as required. Now since $F$ fulfils condition N, for all $e$ we have $F(\rho \circ e) = \rho(F(e))$. Then

\begin{equation}
\begin{aligned}
f_{i',j'}(\pi_{i',j'}(\rho \circ e)) &= \pi_{i',j'}(F(\rho \circ e)) \\
&= \pi_{i',j'}(\rho(F(e)) \\
&= \pi_{i,j}(F(e)) \\
&= f_{i,j}(\pi_{i,j}(e))
\end{aligned}
\end{equation}

But $\pi_{i',j'}(\rho \circ e) = \pi_{i,j}(e)$, and by definition this value ranges over all of $A_{i,j}$, so $f_{i',j'}=f_{i,j}$ as required.
\end{proof}

Thus we can define a single $A = \cup_{i, j} A_{i,j}$ and a single $f: A \rightarrow \set{W,T,L}$ describing the relative SWF between any two candidates. This result now extends to the weight-based SWFs:

\begin{lem}\label{namiia}
Let $F$ be a SWF on a separable ballot domain satisfying MIIA, A and N. Then for any candidates $c_i \neq c_j$ and $c_{i'} \neq c_{j'}$, if $A'_{i,j} \cap A'_{i',j'}$ is non-empty then $A'_{i,j} = A'_{i',j'}$ and $f'_{i,j} = f'_{i',j'}$.
\end{lem}

\begin{proof}
We pick $\vec{\alpha} \in A'_{i,j} \cap A'_{k,\ell}$. We can partition $V$ into measurable sets $V_s$ for all $s \in \supp(\vec{\alpha})$ with $\nu(V_s) = \alpha_s$. Now we select $e \in E$ with $\tau_{i,j}(\pi_{i,j}(e))=\vec{\alpha}$, and $e' \in E$ with $\tau_{k, \ell}(\pi_{k, \ell}(e')) = \vec{\alpha}$, and we let $\vec{\varepsilon} = \tau(e)$ and $\vec{\varepsilon}' = \tau(e')$. Now for each $s \in \supp(\vec{\alpha})$ we partition $V_s$ into $V_r$ for each $r \in \supp(\vec{\varepsilon}) \cap \pi_{i,j}^{-1}(s)$ with $\nu(V_r) = \varepsilon_r$. We define $a \in D_k^V$ by $a(v) = s$ for $v \in V_s$. We also partition $V_s$ into $V'_r$ for each $r \in \supp(\varepsilon')$ with $\nu(V'_r) = \varepsilon'_r$. Then we define $\bar{e}, \bar{e}' \in E$ by $\bar{e}(v) = r$ for $v \in V_r$ and $\bar{e}'(v)=r$ for $v \in V'_r$; these are both in $E$ as it is a ballot domain.

Now $\pi_{i,j}(\bar{e}) = a$ and $\pi_{k, \ell}(\bar{e}') = a$ so $A_{i,j}$ and $A_{k, \ell}$ intersect. Since the domain is separable, we get a permutation $\rho \in \Sym(C)$ with $\rho(c_i) = c_k$, $\rho(c_j) = c_{\ell}$, and $\rho(E)=E$. Applying Lemma \ref{nmiia}, we have $A_{i,j} = A_{k, \ell}$, so $A'_{i,j} = \tau(A_{i,j}) = \tau(A_{k, \ell}) = A'_{k, \ell}$ as required. Moreover, $f_{i,j} = f_{k,\ell}$, so for any choice of $a \in \tau_{i,j}^{-1}(\vec{\alpha})$ we have $f'_{i,j}(\vec{\alpha}) = f_{i,j}(a) = f_{k, \ell}(a) = f'_{k,\ell}(\vec{\alpha})$, so $f'_{i,j} = f'_{k,\ell}$ as required.
\end{proof}

Thus we get a set $A' = \cup_{i,j}A'_{i,j}$, a measuring function $\tau_A: A \rightarrow A'$ and a relative, measure-based SWF $f': A' \rightarrow \set{W,T,L}$ with $f = f' \circ \tau_A$; as always, $f$ uniquely determines $f'$ just as $f'$ uniquely determines $f$.

As in Lemmas \ref{intermediateballots} and \ref{increasingballots}, we now confirm that U and CC are separable domains.

\begin{lem}\label{separableballots}
For either $\mathfrak{B} = \mathfrak{R}$ or $\mathfrak{B} = \mathfrak{C}$, the ballot domain with ballot $\mathfrak{B}$ is a separable domain.
\end{lem}

\begin{proof}
For $\mathfrak{B} = \mathfrak{R}$, any permutation $\rho \in \Sym(C)$ gives $\rho(E)=E$, so we need only construct any $\rho$ with $\rho(c_i)=c_k$ and $\rho(c_j) = c_{\ell}$. Since we know $c_i \neq c_j$ and $c_k \neq c_{\ell}$ the existence of $\rho$ is immediate.

For $\mathfrak{B} = \mathfrak{C}$, if $j=i+1$ then $\pi_{i,j}(\mathfrak{B}) = \set{1,-2}$, whereas if $j=i-1$ then $\pi_{i,j}(\mathfrak{B}) = \set{2,-1}$. These are disjoint, so the domains $A_{i,i+1}$ and $A_{i,i-1}$ are disjoint. On the other hand, if $i-j = k-\ell$ then we take $\rho \in \Sym(C)$ defined by $\rho(c_t) = c_{t+j-i}$; then trivially $\rho(\mathfrak{B}) = \mathfrak{B}$, so $\rho(E) = E$, and $\rho(c_i) = c_j$, $\rho(c_k) = c_{k+j-i} = c_{\ell}$ as required.
\end{proof}

This means that given either condition U or CC, we can assume all results applying to elections on separable, increasing, intermediate ballot domains.

Intuitively it should be the case that receiving votes is good for a candidate. Two conditions are generally used to force this; the Pareto condition and the positive responsiveness condition.

\begin{defn}[Pareto - P]
A SWF satisfies the Pareto condition if, whenever all voters prefer candidate $c_i$ to $c_j$, the same must be the case in the final result. Formally, for $e$ with $\pi_{i,j}(e)(v) > 0$ for all $v$, we have $\pi_{i,j}(F(e))=W$.
\end{defn}

\begin{defn}[Positive Responsiveness - PR]

A SWF satisfies positive responsiveness if whenever $(c_j, c_i)$ is promoted by $(e_2, e_1)$, we have $\pi_{i,j}(F(e_1)) \geq \pi_{i,j}(F(e_2))$.

\end{defn}

The idea of ``promotion'' is difficult to work with. Since we will only be studying ballot elections with increasing ballots and satisfying MIIA, we can check for the PR condition using a simpler property based on the existence of the $f_{i,j}$ functions which MIIA gives. Lemma \ref{prmiia} establishes that for an increasing ballot election satisfying MIIA, this condition is equivalent to PR.

\begin{defn}[Positive Responsiveness given MIIA - PRm]
A SWF satisfying MIIA also satisfies the PR condition if for any two candidates $c_i$ and $c_j$, and for any two relative elections $a_1, a_2 \in A_{i,j}$ where $a_1 \geq a_2$, we have $f_{i,j}(a_1) \geq f_{i,j}(a_2)$.
\end{defn}

\begin{lem}\label{prmiia}
An increasing ballot election with MIIA satisfies PRm if and only if it satisfies PR.
\end{lem}

\begin{proof}
Let $F$ be an increasing ballot election satisfying the MIIA and PRm conditions and let $f_{i,j}$ be its relative SWFs. Suppose that $(c_j, c_i)$ is promoted by $(e_2, e_1)$. Then for any $v$ and for any $c_k \neq c_i, c_j$, letting $r_1 = e_1(v)$ and $r_2 = e_2(v)$, we have:
\begin{enumerate}[(i)]
\item $r_1(c_k) > r_1(c_i) \Rightarrow r_2(c_k) > r_2(c_i)$, so $r_1(c_i) \geq r_2(c_i)$
\item $r_1(c_j) > r_1(c_k) \Rightarrow r_2(c_j) > r_2(c_k)$ so $r_1(c_j) \leq r_2(c_j)$
\end{enumerate}
Thus $r_1(c_i) - r_1(c_j) \geq r_2(c_i)-r_2(c_j)$. Letting $a_1 = \pi_{i,j}(e_1)$ and $a_2 = \pi_{i,j}(e_2)$, we have $a_1(v) \geq a_2(v)$. But this applied to all $v$, so $a_1 \geq a_2$. Then by the PRm condition, $\pi_{i,j}(F(e_1)) = f_{i,j}(a_1) \geq f_{i,j}(a_2) = \pi_{i,j}(F(e_2))$, so the PR condition holds.

Now suppose that $F$ satisfies PR, and take two candidates $c_i$ and $c_j$ and two relative elections in $A_{i,j}$ with $a_1 \geq a_2$. Let $B_{i,j} = \pi_{i,j}(\mathfrak{B})$. Define the function $p: V \rightarrow B_{i,j}^2$ by $p(v) = (a_1(v),a_2(v))$. Note that $p$ is measurable, as the cartesian product measure of $B_{i,j} \times B_{i,j}$ is just the discrete measure (because $B_{i,j}$ is finite). For any pair $(b_1,b_2)$ in the image of $p$, note that $b_1 \geq b_2$, and choose any $r(b_1,b_2) \in \pi_{i,j}^{-1}(b_1)$.

Then, since $\mathfrak{B}$ is an increasing ballot, let $r'(b_1,b_2) \in \pi_{i,j}^{-1}(b_2)$ be such that $(c_j, c_i)$ is promoted by $(r'(b_1,b_2),r(b_1,b_2))$. Then the functions $r$ and $r'$ are both functions with finite domains, and therefore measurable. Define $e$ by $e(v) = r(a_1(v),a_2(v))$ and $e'$ by $e'(v) = r'(a_1(v),a_2(v))$; these are both measurable functions, so they are in $E$ as it is a ballot election. Then $(c_j,c_i)$ is promoted by $(e',e)$, so by PR we have $\pi_{i,j}(a_1) = \pi_{i,j}(F(e_1)) \geq \pi_{i,j}(F(e_2)) = \pi_{i,j}(a_2)$, giving condition PRm.
\end{proof}

Note that the Pareto condition is often much weaker than the Positive Responsiveness condition, as it says nothing except in the very extreme case of unanimity among voters when comparing two candidates. On the other hand, the Positive Responsiveness condition does not necessarily imply the Pareto principle, as PR doesn't rule out the possibility of $F$ always returning the same ranking. Even imposing the neutrality condition, $F$ could always return a $k$-way tie, whereas the Pareto principle precludes this.

The following condition is sometimes included in the PR condition but we will consider it separately.

\begin{defn}[Intermediate Value - IV]

If for two elections $e_1$ and $e_2$ and two candidates $c_i$ and $c_j$, $(c_j, c_i)$ is promoted by $(e_2, e_1)$, and we have $\pi_{i,j}(F(e_1)) = W$ whereas $\pi_{i,j}(F(e_2))=L$, then there exists an ``intermediate'' election $e'$ such that $(c_j, c_i)$ is promoted by both $(e', e_1)$ and $(e_2, e')$, and $\pi_{i,j}(F(e'))=T$.
\end{defn}

The IV condition essentially says that we can continuously improve a candidate's position relative to another candidate until their positions cross in the final ranking, and at that moment they will tie. We would only expect this to be the case in a world with a continuum of possible ballots; in other words, when $|V| = \infty$. We will only consider the IV condition in these cases.

Again, we can verify the IV condition for SWFs on increasing, intermediate domains with MIIA using the simpler IVm condition:

\begin{defn}[Intermediate Value given MIIA - IVm]
A SWF satisfying MIIA also satisfies the IV condition if for any two candidates $c_i$ and $c_j$, and for any $a_1 \geq a_2 \in A_{i,j}$ such that $f_{i,j}(a_1) = W$ and $f_{i,j}(a_2)=L$, there exists $a' \in A_{i,j}$ with $a_1 \geq a' \geq a_2$ and $f_{i,j}(a')=T$.
\end{defn}

\begin{lem}\label{ivmiia}
An increasing, intermediate ballot election with MIIA satisfies IV if and only if it satisfies IVm.
\end{lem}

\begin{proof}
Let $F$ be a SWF satisfying the MIIA and IV conditions, let $f_{i,j}$ be its relative SWFs and let $a_1 \geq a_2$ be two vectors in $A_{i,j}$ for some $c_i$ and $c_j$, with $f_{i,j}(a_1)=W$ and $f_{i,j}(a_2)=L$. Let $B_{i,j} = \pi_{i,j}(\mathfrak{B})$. For $(b_1, b_2) \in B_{i,j}^2$ pick $r_1 \in \pi_{i,j}^{-1}(b_1)$, and use the definition of an increasing ballot to pick $r_2 \in \pi_{i,j}^{-1}(b_2)$ such that $(c_j,c_i)$ is promoted by $(r_2,r_1)$. Now both $r_1$ and $r_2$ are given by functions on $B_{i,j}^2$, and since these are functions on finite domains, they are measurable.

Now define $e_1$ by $e_1(v) = r_1(a_1(v),a_2(v))$ and $e_2$ by $e_2(v) = r_2(a_1(v),a_2(v))$. Since $F$ has a ballot domain, $e_1, e_2 \in E$. Moreover, $(c_j,c_i)$ is promoted by $(e_2, e_1)$, $\pi_{i,j}(e_1)=a_1$ and $\pi_{i,j}(e_2)=a_2$. Thus $\pi_{i,j}(F(e_1))=W$ and $\pi_{i,j}(F(e_2))=L$. Now applying the IV condition, we find that there exists $e'$ such that $(c_j,c_i)$ is promoted by both $(e_2, e')$ and $(e',e_1)$, and $\pi_{i,j}(F(e'))=T$. Setting $a' = \pi_{i,j}(e')$, we have $f_{i,j}(a')=T$, and $a_1 \geq a' \geq a_2$, giving condition IVm.

Now let $F$ be a SWF on an intermediate ballot domain satisfying the MIIA and IVm conditions and let $f_{i,j}$ be its relative SWFs. Suppose that $(c_j, c_i)$ is promoted by $(e_2,e_1)$ and that $\pi_{i,j}(F(e_1)) = W$ and $\pi_{i,j}(F(e_2)) = L$. Now let $\pi_{i,j}(e_1)=a_1$ and $\pi_{i,j}(e_2)=a_2$. Since $(c_j, c_i)$ is promoted by $(e_2,e_1)$, we know that $a_1 \geq a_2$. Now we have $f_{i,j}(a_1) = W$ and $f_{i,j}(a_2) = L$. Thus by IVm there exists $a' \in A_{i,j}$ with $a_1 \geq a' \geq a_2$ and $f_{i,j}(a')=T$.

Meanwhile, for each triple $(r_+, m, r_-) \in \mathfrak{B} \times B_{i,j} \times \mathfrak{B}$ with $(c_j,c_i)$ promoted by $(r_-,r_+)$ we carry out the following procedure: if $m = \pi_{i,j}(r_+)$ we set $r' = r_+$, and if $m=\pi_{i,j}(r_-)$ we set $r' = r_-$; otherwise, we apply the definition of the intermediate ballot to find $r'$ such that $(c_j,c_i)$ is promoted by both $(r_-,r')$ and $(r',r_+)$, and where $\pi_{i,j}(r')=m$. This gives a function on finite domain $r': \mathfrak{B} \times B_{i,j} \times \mathfrak{B} \rightarrow \mathfrak{B}$.

Now for each voter $v$, we have $r_+ = e_1(v)$ and $r_- = e_2(v)$ with $(c_j, c_i)$ promoted by $(r_+, r_-)$; and

\begin{equation}
\pi_{i,j}(r_+) = a_1(v) \geq a'(v) \geq a_2(v) = \pi_{i,j}(r_-)
\end{equation}

so we can define $e'(v) = r'(e_1(v), a'(v), e_2(v))$. This is a measurable function, so $e' \in E$ (as $E$ is a ballot election). Clearly $(c_j,c_i)$ is promoted by both $(e_2,e')$ and $(e',e_1)$, and $\pi_{i,j}(e')=a'$. Also, $\pi_{i,j}(F(e')) = f_{i,j}(a') = T$, so we have condition IV.
\end{proof}

These results taken together allow us to simplify our analysis of SWFs that are sufficiently well-behaved. On the other hand, it is often possible to start with relative SWFs and check that they corresponds to a full SWF. To prepare for this construction, we define the following ``consistent'' sets of values we can get for the relative rankings of three candidate:

\begin{defn}[Consistent multisets]
A multiset of cardinality three, with elements drawn from $\set{W,T,L}$, is consistent if it is one of the following:
\begin{enumerate}
\item $\set{W,W,L}$
\item $\set{W,T,L}$
\item $\set{W,L,L}$
\item $\set{T,T,T}$
\end{enumerate}

Otherwise the multiset is ``inconsistent''.
\end{defn}

If we take a weak ordering $\succ$ on a set of three candidates $c_i, c_j, c_k$, then the consistent multisets are the ones that could arise from finding $\pi_{i,j}(\succ)$, $\pi_{j,k}(\succ)$ and $\pi_{k,i}(\succ)$, as we will verify in the proof of the following Lemma.

\begin{lem}\label{relativetofull}
Given a family of relative SWFs $f_{i,j}$ for each $c_i \neq c_j \in C$, there exists an SWF $F$ with $\pi_{i,j} \circ F = f_{i,j} \circ \pi_{i,j}$ for all $i \neq j$ if and only if for all $e \in E$:
\begin{enumerate}
\item \label{relativetofull:twoway} for all $c_i \neq c_j$ we have $f_{i,j}(\pi_{i,j}(e)) = -f_{j,i}(\pi_{j,i}(e))$, and
\item \label{relativetofull:threeway} for all $c_i \neq c_j \neq c_k$, the multiset $\set{f_{i,j}(\pi_{i,j}(e)), f_{j,k}(\pi_{j,k}(e)), f_{k,i}(\pi_{k,i}(e))}$ is consistent.
\end{enumerate}
\end{lem}

\begin{proof}
Suppose that $F$ does exist with $\pi_{i,j} \circ F = f \circ \pi_{i,j}$ for all $c_i \neq c_j$. Take any $e \in E$ and $c_i \neq c_j$. Now let $F(e) = \succ$. If $c_i \succ c_j$, then $f(\pi_{i,j}(e)) = \pi_{i,j}(F(e))=W$, while $f(\pi_{j,i}(e)) = \pi_{j,i}(F(e)) = L$, satisfying condition \ref{relativetofull:twoway}; but the same goes if $c_j \succ c_i$, or if $c_i \nsucc c_j$ and $c_j \nsucc c_i$.

Similarly, taking $c_i \neq c_j \neq c_k$, and let $F(e) = \succ$. Now we consider two cases. If all three candidates are tied, then the multiset described in condition \ref{relativetofull:threeway} is $\set{T,T,T}$, which is consistent. Otherwise, there must be one maximal or minimal candidate of the three; assume without loss of generality that this is $c_i$. Then one of $\pi_{i,j}(\succ)$ and $\pi_{k,i}$ is $W$ and the other is $L$; so the multiset is one of $\set{W,W,L}$, $\set{W,T,L}$ and $\set{W,L,L}$; these are all consistent.

On the other hand, suppose that conditions \ref{relativetofull:twoway} and \ref{relativetofull:threeway} are met. Then we define $\succ = F(e)$ as follows. We construct a graph  $G'$ with vertices corresponding to $C$ the set of candidates, and for each pair $c_i \neq c_j$ we add an edge between $c_i$ and $c_j$ when $f_{i,j}(\pi_{i,j}(e)) = T$. Note that these edges can be undirected, as condition \ref{relativetofull:twoway} ensures that if $f_{i,j}(\pi_{i,j}(e))=T$ then $f_{j,i}(\pi_{j,i}(e)) = -T = T$. Now if in this graph $c_i \sim c_j$ and $c_j \sim c_k$, then by condition \ref{relativetofull:threeway} we must have $c_i \sim c_k$, because the only consistent multiset containing $T$ twice is $\set{T,T,T}$; and therefore the graph is a disjoint collection of cliques. We call these sets $C_1, C_2 \dots C_t$ for some $t$. 

Now we construct a new directed graph $G$ on vertices $V = \set{C_i}_i$, and we add an edge from $C_i$ to $C_j$ if, for some $c_i \in C_i$ and $c_j \in C_j$, we have $f_{i,j}(\pi_{i,j}(e)) = W$. This definition is independent of our choice of $c_i$ and $c_j$, because if we apply condition \ref{relativetofull:threeway} to $c_i$, $c_{i'}, c_j$ with $c_i, c_{i'} \in C_i$, we have $f_{i,i'}(\pi_{i,i'}(e)) = T$ and $f_{i,j}(\pi_{i,j}(e)) \neq T$, so the multiset is $\set{W,T,L}$; then $f_{i,j}(\pi_{i,j}(e))= -f_{j,i'}(\pi_{j,i'}(e))$, so by condition \ref{relativetofull:twoway} we have $f_{i,j}(\pi_{i,j}(e)) = f_{i',j}(\pi_{i',j}(e))$; applying this again with $c_{i'}$, $c_j$ and $c_{j'}$ we get $f_{i,j}(\pi_{i,j}(e)) = f_{i',j'}(\pi_{i',j'}(e))$.

Now $(C_i, C_j)$ is in $G$ if and only if $(C_j,C_i)$ is not; they cannot both be edges due to condition \ref{relativetofull:twoway} but if neither is an edge then $c_i \nsucc c_j$ and $c_j \nsucc c_i$, so $C_i$ and $C_j$ were connected in $G'$ and they formed one clique after all. Therefore $G$ is a tournament. Suppose there is a cycle in the tournament, $(C_1,C_2, \dots, C_s)$ with $(C_i, C_{i+1})$ all edges. Then there is an edge from $C_1$ to $C_3$ (or else, taking $c_1 \in C_1$, $c_2 \in C_2$ and $c_3 \in C_3$, and applying condition \ref{relativetofull:threeway} to $(c_1, c_2, c_3)$ we get one of $\set{W,W,W}$ or $\set{W,W,T}$, both of which are inconsistent), so we get a smaller cycle. We repeat this until we get a 3-cycle, which corresponds to the inconsistent multiset $\set{W,W,W}$, a contradiction. So instead, there is a total order between the vertices of $G$.

Now we define $\succ$ based on these two graphs. We say that $c_i \nsucc c_i$. For $c_i \neq c_j$, if $c_i \sim c_j$ in $G'$, then the result between them is a tie, so we write $c_i \nsucc c_j$ and $c_j \nsucc c_i$. Otherwise, they are in different connected components; say $c_i \in C_i$ and $c_j \in C_j$. Then we write $c_i \succ c_j$ if $(C_i,C_j)$ is an edge and $c_j \succ c_i$ if $(C_j,C_i)$ is an edge. Clearly $c_i \succ c_j$ if and only if $f_{i,j}(\pi_{i,j}(e))=W$, and $c_j \succ c_i$ if and only if $f_{i,j}(\pi_{i,j}(e))=L$; neither occur if and only if $f_{i,j}(\pi_{i,j}(e))=T$. Moreover, $\succ$ is a weak ordering: if $c_i \succ c_j$ and $c_j \succ c_k$ then $c_i \succ c_k$; and if $c_i \nsucc c_j$ and $c_j \nsucc c_k$ then $c_i \nsucc c_k$.
\end{proof}

\subsection{Borda Rules}\label{section:borda}

We now define the (weighted) Borda rules, which are the obvious examples of SWFs meeting the MIIA condition.

\begin{defn}[(Weighted) Borda rule]

For a finite set of voters $V$ with weights $(w_v)_v$, or for an infinite set of voters $(V, \Sigma, \nu)$ with a measurable weight function $w: V \rightarrow \mathbbm{R}$, the weighted Borda rule $B_w$ is defined as follows. We define

\begin{equation}
\begin{aligned}
b_w(c_i,e) &= \int_V e(v)(c_i)w(v)\diff \nu(v) \quad & |V| = \infty \\
b_w(c_i,e) &= \sum_{v \in V} e(v)(c_i)w_v \quad & |V| < \infty
\end{aligned}
\end{equation}

Now we let $B_w = \Phi(b_w)$.

\end{defn}

Now $\pi_{i,j}(B_w(e)) = \varphi(b_w(c_i,e) - b_w(c_j,e))$, so we let

\begin{equation}\label{eq:relativeborda}
\begin{aligned}
d_w(c_i,c_j,e) &= b_w(c_i,e)-b_w(c_j,e) \\
&= \int_V [e(v)(c_i)-e(v)(c_j)] w(v) \diff \nu(v) \\
&= \int_V \pi_{i,j}(e)(v) w(v) \diff \nu(v) \quad |V| = \infty
\end{aligned}
\end{equation}

The equivalent for $|V|$ finite is

\begin{equation}
d_w(c_i,c_j,e) = \sum_{v \in V} \pi_{i,j}(e)(v) w_v
\end{equation}

Thus $d_w$ is a function of $\pi_{i,j}(e)$, and we can write $d_w(c_i,c_j,e) = d_w(\pi_{i,j}(e))$. Then $B_w$ satisfies MIIA with $f_{i,j} = \varphi \circ d_w$. Moreover, $b_w$ is defined symmetrically for all $c_i \in C$, so $B_w$ satisfies the neutrality condition.

We call a Borda rule ``positive'' if $w > 0$ across $V$ and ``non-negative'' if $w \geq 0$ across $V$, and so on for ``negative'' and ``non-positive''. Equation \ref{eq:relativeborda} means that a non-negative Borda rule satisfies PRm, and that a non-negative Borda rule with $\int_V w(v) \diff \nu(v) >0$ satisfies P.

Moreover, on a ballot domain with $V$ infinite, we have IVm. Indeed, we will construct $(V_r)_r$ to be an infinite family of measurable subsets of $V$ with $\nu(V_r)=r$ and $V_r \subset V_s$ whenever $r<s$. We let $V_0 = \emptyset$ and $V_1=V$, and then labelling the rationals in $(0,1)$ with the natural numbers, we define each $V_{q_n}$ in turn to contain all $V_{q_m}$ with $m<n$ and $q_m<q_n$, and to be contained in all $V_{q_m}$ with $m<n$ and $q_m>q_n$; once $V_q$ is defined for all rational $q$ we let $V_r$ for irrational $r$ be the union of $V_q$ for all $q<r$; clearly this satisfies our criteria.

Then for relative elections $a_-$ and $a_+$ with $a_- < a_+$ and $\varphi \circ d_w(a_-) = L$ and $\varphi \circ d_w(a_+)=W$, we define $a_r$ by $a_r(v) = a_+(v)$ on $V_r$ and $a_r(v) = a_-(v)$ elsewhere. Now $d_w(a_r)$ is a continuous function on $r$, with $d_w(a_0) < 0$ and $d_w(a_1) > 0$ (as these values map to $L$ and $W$ respectively under $\varphi$). Thus by the intermediate value theorem there exists $r$ with $d_w(a_r)=0$, and we have $a_- < a_r < a_+$ and $\varphi \circ d_w(a_r) = T$ as required.

If $w$ is constant across all of $V$ then the Borda rule is ``unweighted''. Since scaling $b_w$ by a positive scalar doesn't affect $\Phi$, there are only three such rules: the positive unweighted Borda rule with $w\equiv 1$, the negative unweighted Borda rule with $w \equiv -1$, and the ``tie rule'' with $w \equiv 0$, in which case the SWF always returns a tie between all candidates in $C$.

For any unweighted Borda rule, we have

\begin{equation}
\begin{aligned}
b_w(c_i,e) &=  \sum_{r \in \mathfrak{R}} \int_{v \in e^{-1}(r)} w \cdot r(c_i) \diff \nu(v) \\
&= w \sum_{r \in \mathfrak{R}} r(c_i) \nu(e^{-1}(r)) \\
&= w \sum_{r \in \mathfrak{R}} r(c_i) \vec{\varepsilon}_r
\end{aligned}
\end{equation}

Here $\vec{\varepsilon} = \tau(e)$. Therefore $b_w$ is a function of $\tau(e)$, so $B_w$ satisfies the anonymity condition. Again, this also applies to the finite case.

Between Borda and non-Borda SWFs we can also consider ``weakly Borda'' SWFs, which agree with a Borda rule $B_w$ (for $w \not\equiv 0$) whenever it is decisive but are sometimes able to break ties.

\begin{defn}[Weakly Borda]
A SWF $F$ is weakly Borda if there exists a Borda rule $B_w$ for $w \not\equiv 0$ (i.e. not the tie rule) such that for all $e$, and for two candidates $c_i$ and $c_j$, we have $\pi_{i,j}(B_w(e)) \neq T \Rightarrow \pi_{i,j}(F(e)) = \pi_{i,j}(B_w(e))$. We also say that $F$ is weakly Borda with respect to $B_w$.
\end{defn}

Note that any Borda rule except for the tie rule is weakly Borda with respect to itself. We call a rule ``strongly non-Borda'' if it is not weakly Borda or the tie rule.

The anonymity condition A should mean that the only sensible Borda weights to take are unweighted ones. We formalise this in Lemma \ref{weightinglemma}, which is useful for demonstrating that a SWF is not Borda; it allows us to check only the unweighted Borda rules, and not to worry about all the weighted Borda rules.

\begin{lem}\label{weightinglemma}
For an election domain $E$ containing a ballot domain with ballot $\mathfrak{B}$, if a SWF $F$ fulfils condition A and is weakly Borda with respect to a Borda rule $B_w$ then it is weakly Borda with respect to an unweighted Borda rule $B$.
\end{lem}

\begin{proof}
We consider the case of $V$ infinite; for $V$ finite, an identical argument works, letting $\nu(X)$ be $|X|/|V|$ for any set of voters $X$.

Consider the function $w$. For brevity, we denote the interval $(t,\infty)$ as $U_t$ and the interval $(-\infty,t)$ as $L_t$. We then denote $W^+_t = w^{-1}(U_t)$ and $W^-_t = w^{-1}(L_t)$, and we write $\nu(W^{\varepsilon}_t) = w^{\varepsilon}_t$ for $\varepsilon = \pm$.

Furthermore, let $r$ and $r'$ be two distinct votes in $\mathfrak{B}$. There must be two candidates $c_i$ and $c_j$ for which $r(c_i)>r(c_j)$ but $r'(c_i)<r'(c_j)$ (otherwise $r$ and $r'$ give identical total orders of $C$ and are therefore equal). Let $s = r(c_i)-r(c_j)$ and $t = r'(c_j) - r'(c_i)$, and let $\kappa = s/(s+t)$.

Now $w^-_{\lambda}$ is an increasing function of $\lambda$; we take $\lambda_-$ to be the infimum of all values for which $w^-_{\lambda} \geq \kappa$. Now $w^-_{\lambda_-} \leq \kappa$, as it is the supremum of $w^-_{\lambda}$ for all $\lambda < \lambda_-$ (since $L_{\lambda_-}$ is the union for all $\lambda<\lambda_-$ of $L_{\lambda}$), and for $\lambda < \lambda_-$ we have $w^-_{\lambda} < \kappa$. On the other hand, $U_r = \cup_{r' > r} L_r^c$ so $w^+_{\lambda_-} = \sup_{\lambda > \lambda_-} 1 - w^-_{\lambda}$. But for all $\lambda > \lambda_-$ we have $w^-_{\lambda} \geq \kappa$, so $w^+_{\lambda_-} \leq 1 - \kappa$.

Now, if $w^-_{\lambda_-} = \kappa$ then we set $S_- = W^-_{\lambda_-}$ and $T_- = V \setminus S_-$. On the other hand, if $w^-_{\lambda_-} < \kappa$ then we have $\nu(w^{-1}(\lambda_-)) = 1 - w^-_{\lambda_-} - w^+_{\lambda_-} \geq \kappa - w^-_{\lambda_-}$. Now we take a measurable subset $X$ of $w^{-1}(\lambda_-)$ of measure $\kappa - w^-_{\lambda_-}$, and we set $S_- = X \cup W^-_{\lambda_-}$ and $T_- = V \setminus S_-$.

Now we have partitioned $V$ into $S_-$ and $T_-$ with $\nu(S_-) = \kappa$, $w \leq \lambda_-$ on $S_-$ and $w \geq \lambda_-$ on $T_-$. Since $E$ is a ballot election, we can find $e_- \in E$ defined by $e_-(v) = r'$ for $v \in S_-$ and $e_-(v) = r$ for $v \in T_-$. Now for the election $e_-$ we have 

\begin{equation}
\begin{aligned}
d_w(c_i, c_j, e_-) &= \int_{T_-} s w(v) \diff \nu(v) - \int_{S_-} t w(v) \diff \nu(v) \\
&\leq \lambda_- (1 - \kappa) s - \lambda_- \kappa t \\
&= 0
\end{aligned}
\end{equation}

If $d_w(c_i, c_j, e_-) < 0$ then $\pi_{i,j}(B_w(e_-)) = L$ and, since $F$ is weakly Borda, $\pi_{i,j}(F(e_-)) = L$.

Similarly, we can construct $\lambda_+$, and then $S_+$ and $T_+$ partitioning $V$ such that $\nu(S_+) = \kappa$, $w \geq \lambda_+$ on $S_+$ and $w \leq \lambda_+$ on $T_+$. Then setting $e_+$ to be the election with $e_+(v) = r'$ for $v \in S_+$ and $e_+(v) = r$ for $v \in T_+$, we find that

\begin{equation}
\begin{aligned}
d_w(c_i,c_j,e_+) &= \int_{T_+} s w(v) \diff \nu(v) - \int_{S_+} t w(v) \diff \nu(v) \\
&\geq \lambda_+ (1 - \kappa) s - \lambda_+ \kappa t \\
&= 0
\end{aligned}
\end{equation}

Again, if $d_w(c_i, c_j, e_+) > 0$ then $\pi_{i,j}(B_w(e_+)) = W$ and, since $F$ is weakly Borda, $\pi_{i,j}(F(e_+)) = W$. But since $\tau(e_-) = \tau(e_+)$, and $F$ satisfies anonymity, $\pi_{i,j}(F(e_-)) = \pi_{i,j}(F(e_+))$. Therefore for at least one election $e$ out of the two elections $e_-$ and $e_+$ we have $d_w(c_i, c_j, e) = 0$. Suppose without loss of generality that $e = e_-$. Then $\int_{T_-} w(v) \diff \nu(v) = (1 - \kappa)\lambda_-$, so $w = \lambda_-$ on all but a zero-measure subset $T_0$ of $T$. Similarly $w = \lambda_-$ on all but a zero-measure subset $S_0$ of $S$. Then for all $e\in E$ and for all $c_i \in C$, $b_w(c_i) = b_{\lambda_-}(c_i)$ (we are abusing notation to write $\lambda_-$ for the constant weight function $w \equiv \lambda_-$). Therefore $F$ is weakly Borda with respect to the unweighted Borda rule with sign $\sgn(\lambda_-)$.
\end{proof}

We now proceed to the main results of this paper.

\section{Unrestricted domain and anonymity}\label{UA}

We begin with the most critical construction, which demonstrates that on the unrestricted domain, even the combination of all of the conditions listed in section \ref{section:conditions} is not enough to guarantee the Borda rule. We give a construction for the case of $V$ infinite, and then use this construction to generate more constructions for cases of $V$ finite.

\begin{thm}\label{IUA}
There exists a strongly non-Borda SWF on $(V, \nu)$ and $C = \set{c_1,c_2,c_3}$ satisfying U, MIIA, A, N, P, PR and IV.
\end{thm}

Note that by Theorem \ref{ICA}, proved in section \ref{CA} below, the restriction of the SWF to either of the two Condorcet cycle domains will give the Borda rule. We can think of two elections running in parallel: one between all the voters in the first Condorcet cycle, $e^{-1}(\mathfrak{C}) \subset V$, and a second between voters in the second Condorcet cycle, $e^{-1}(\mathfrak{R} \setminus \mathfrak{C}) = V \setminus e^{-1}(\mathfrak{C})$. By giving a pathological rule for combining the Borda scores of these two elections, we will find a non-Borda election on the unrestricted domain.

\begin{proof}[Proof of Theorem \ref{IUA}]

We consider the following function $G': C \times E' \rightarrow \mathbbm{R}$:

\begin{equation} \label{iua:fullscore}
G'(c_i, \vec{\varepsilon}) = X\sum_{r \in \mathfrak{R}} r(c_i) \varepsilon_r + \sum_{r \in \mathfrak{R}}\sum_{s \in C_3 \circ r} r(c_i)\varepsilon_r \varepsilon_s
\end{equation}

Here $C_3$ is the group of cyclic (i.e. even) permutations on $C$. $X$ will be determined later.

We can define $G: C \times E \rightarrow \mathbbm{R}$ by $G(c,e) = G'(c, \tau(e))$. Then we let $F = \Phi(G)$. Now the relative result between any $c_i$ and $c_j$ is given by

\begin{equation}
\pi_{i,j}(F(e)) = \varphi(G'(c_i,\vec{\varepsilon})-G'(c_j,\vec{\varepsilon}))
\end{equation}

Now $F$ satisfies condition U by construction.

\begin{lem}
$F$ satisfies MIIA, N and A.
\end{lem}

\begin{proof}
It suffices to show that $F$ has a corresponding weight-based relative SWF $f$ (neutrality is given by the fact that this function will not depend on the choice of $c_i$ and $c_j$). If $G'(c_i,\vec{\varepsilon})-G'(c_j,\vec{\varepsilon}) = g(\pi_{i,j}(\vec{\varepsilon}))$ for some $g$, then $f' = \varphi \circ g$ is the weight-based relative SWF corresponding to $F$. Indeed, writing $\vec{\alpha} = \pi_{i,j}(\vec{\varepsilon})$, we have

\begin{equation}
\begin{aligned}
G'(c_i,\vec{\varepsilon})-G'(c_j,\vec{\varepsilon}) =& X\sum_{r \in \mathfrak{R}} r(c_i) \varepsilon_r + \sum_{r \in \mathfrak{R}}\sum_{s \in C_3 \circ r} r(c_i)\varepsilon_r \varepsilon_s \\
&- X\sum_{r \in \mathfrak{R}} r(c_j) \varepsilon_r - \sum_{r \in \mathfrak{R}}\sum_{s \in C_3 \circ r} r(c_j)\varepsilon_r \varepsilon_s \\
=& X\sum_{r \in \mathfrak{R}} (r(c_i)-r(c_j)) \varepsilon_r + \sum_{r \in \mathfrak{R}}\sum_{s \in C_3 \circ r} (r(c_i)-r(c_j))\varepsilon_r \varepsilon_s \\
=& X \sum_{t \in D_3} \left( t \sum_{r: \pi_{i,j}(r)=t} \varepsilon_r \right) + \sum_{t \in D_3}\left(t \sum_{r: \pi_{i,j}(r)=t} \varepsilon_r \sum_{s \in C_3 \circ r} \varepsilon_s \right) \\
=& X \sum_{t \in D_3} t \alpha_t + \sum_{t \in D_3} t \alpha_t \left(\sum_{s: 3 \mid \pi_{i,j}(s) - t} \varepsilon_s \right) \\
=& X \sum_{t \in D_3} t\alpha_t + \sum_{t \in D_3} t \alpha_t \left(\sum_{u: 3 \mid (t-u)} \alpha_u\right)
\end{aligned}
\end{equation}

The penultimate inequality is because when applying an even permutation to $r \in \mathfrak{R}$, we maintain the cyclic order of the three candidates, so a gap between two candidates of $1$ or $-2$ must remain a gap of $1$ or $-2$, and the same goes for $-1$ and $2$. The final expression is a function only of $\vec{\alpha}$, so we can define it as

\begin{equation}\label{iua:relativescore}
g(\vec{\alpha}) = (X+\alpha_1 + \alpha_{-2})(\alpha_1 - 2\alpha_{-2}) + (X+\alpha_2 + \alpha_{-1})(2\alpha_2 - \alpha_{-1})
\end{equation}

giving $f' = \varphi \circ g$ a weight-based relative SWF as required.
\end{proof}

It remains to show that PR, P and IV are also satisfied.

\begin{lem} $F$ satisfies condition PR. \end{lem}
\begin{proof}
By Lemma \ref{prmiia}, we need only check condition PRm, which states that as $\vec{\alpha}$ increases in the partially ordered set $A'$, the value of $f(\vec{\alpha})$ cannot decrease.

For any $\vec{\alpha} < \vec{\beta}$ in $A'$, we can reach $\vec{\beta}$ from $\vec{\alpha}$ by a series of at most three steps. In the first one, we replace $\alpha_{-2}$ with the smaller $\beta_{-2}$ and add $\alpha_{-2} - \beta_{-2}$ to $\alpha_{-1}$. We know that $\beta_{-2} \leq \alpha_{-2}$ as $\vec{\beta}$ majorises $\vec{\alpha}$, so

\begin{equation}
1 - \beta_{-2} = \beta_2 + \beta_1 + \beta_{-1} \geq \alpha_2 + \alpha_1 + \alpha_{-1} = 1 - \alpha_{-2}
\end{equation}

But now we are left with two vectors of equal total weight on indices $2, 1, -1$, where $\beta_2 + \beta_1 \geq \alpha_2 + \alpha_1$, so again we can replace $\alpha_{-1} + \alpha_{-2}-\beta_{-2}$ with the smaller $\beta_{-1}$, and increase $\alpha_1$ to $\alpha_1+\alpha_{-1}+ \alpha_{-2} - \beta_{-2}-\beta_{-1}$. Finally we decrease this to $\beta_1$ and increase $\alpha_2$ to $\beta_2$.

Carrying out each of these stages continuously, at each point we are replacing $\vec{\alpha}$ with one of

\begin{equation}\label{iua:shifts}
\begin{aligned}
(\alpha_2+\delta, \alpha_1-\delta, \alpha_{-1}, \alpha_{-2}) \\
(\alpha_2, \alpha_1+\delta, \alpha_{-1}-\delta, \alpha_{-2}) \\
(\alpha_2, \alpha_1, \alpha_{-1}+\delta, \alpha_{-2}-\delta)
\end{aligned}
\end{equation}

for $\delta$ arbitrarily small. These respectively increase $g(\vec{\alpha})$ by

\begin{equation}
\begin{aligned}
\delta(X+4\alpha_2 + \alpha_{-1} + 2\alpha_1 - \alpha_{-2})+3\delta^2 \\
\delta(X-\alpha_2 - \alpha_{-2} +2\alpha_1) \\
\delta(X+\alpha_2-2\alpha_{-1}+\alpha_1+4\alpha_{-2})-3\delta^2
\end{aligned}
\end{equation}

For $X \geq 2$, the coefficients of $\delta$ on the first two lines are both posiitive, and the entire expressions are positive. On the third line, $X \geq 2$ guarantees that the coefficient of $\delta$ is non-negative, and it is zero if and only if $\alpha_{-1}=1$. For $\alpha_{-1} \neq 1$, we pick $\delta$ small enough for the linear term to dominate the quadratic term, so the expression is positive. On the other hand, if $\alpha_{-1}=1$ then we cannot apply the third variation of $\vec{\alpha}$, as it replaces $\alpha_{-1}$ with $\alpha_{-1}+\delta > 1$, which is impossible.

Thus, transforming $\vec{\alpha}$ to $\vec{\beta}$ in small enough steps, we find that $g(\vec{\beta})>g(\vec{\alpha})$, so $f(\vec{\beta}) \geq f(\vec{\alpha})$, and the PRm condition holds.
\end{proof}

Now we fix $X=2$ and we quickly address the Pareto principle:

\begin{lem} $F$ satisfies condition P. \end{lem}
\begin{proof}
The Pareto condition requires that if $\alpha_2 + \alpha_1 = 1$ then $g(\vec{\alpha}) > 0$. The expression for $g$ given that $\alpha_{-1}=\alpha_{-2}=0$ is $\alpha_1(X+\alpha_1) + 2\alpha_2(X+\alpha_2)$. Since $X\geq 0$ this is bounded below by $\alpha_1^2+\alpha_2^2 \geq (\alpha_1 + \alpha_2)^2/2= 1/2$, so the SWF satisfies condition P.
\end{proof}

\begin{lem} $F$ satisfies condition IV. \end{lem}

\begin{proof}
Again, by Lemma \ref{ivmiia} we need only check for condition IVm, which states that for $a_+, a_- \in A$ with $a_+ > a_-$, $f(a_+) = W$ and $f(a_-)=L$, there is an intermediate $a$ with $a_+ > a > a_-$ and $f(a)=T$. In terms of vectors in $A'$, we require for any $\vec{\alpha}_+$ with $g(\vec{\alpha}_+) > 0$ and $\vec{\alpha}_-$ with $g(\vec{\alpha}_-)<0$ where $\vec{\alpha}_+ > \vec{\alpha}_-$ (in other words, where $\vec{\alpha}_+$ is reached from $\vec{\alpha}_-$ by a sequence of the steps listed in \ref{iua:shifts}) that there exists an intermediate vector $\vec{\alpha}$ with $g(\vec{\alpha})=0$ and $\vec{\alpha}_+ > \vec{\alpha} > \vec{\alpha}_-$. But this is clear, because varying $\delta$ from $0$ causes $g$ to vary continuously, so by the Intermediate Value  Theorem there must be a point at which $g=0$, and this gives the required $\vec{\alpha}$.
\end{proof}

Thus $F$ is an SWF meeting all of our conditions. All that remains is to show that it is non-Borda.

\begin{lem} $F$ is strongly non-Borda. \end{lem}

\begin{proof}
Note that $g(\vec{\alpha})$ is a polynomial in the entries of $\vec{\alpha}$ and does not contain $b(\vec{\alpha}) = 2\alpha_2 + \alpha_1-\alpha_{-1}-2\alpha_{-2}$ as a factor; therefore $g$ is a non-zero polynomial on the domain given by $b=0$ and $\vec{\alpha} \geq 0$. There is a non-zero on the interior of this domain, representing a point $\vec{\alpha}_0$ where all $\alpha_t$ are positive, $g(\vec{\alpha}) \neq 0$ and $b(\vec{\alpha})=0$. Assuming without loss of generality that $g>0$, we set $\alpha'_j = \alpha_j$ for $j <0$, and write $\alpha'_2 = \alpha_2-\delta$ and $\alpha'_1 = \alpha_1+\delta$ with $\delta < \alpha_2$ small enough not to affect the sign of $g$. Then $b(\vec{\alpha}') <0$ while $g(\vec{\alpha}')>0$.

Clearly this value for $\vec{\alpha}$ is given by $\pi_{1,2} \circ \tau(e)$ for some election $e$, because $A' = \Image(\pi_{1,2} \circ \tau)$ is all non-negative vectors $\vec{\alpha}$ with $\sum_i \alpha_i = 1$. Therefore the positive unweighted Borda rule $B(e)$ and our SWF $F(e)$ give opposing rankings for $c_1$ and $c_2$, so $F$ is neither the positive unweighted Borda rule nor weakly Borda with respect to it.

An identical argument, setting $\alpha'_2 = \alpha_2+\delta$ and $\alpha'_1 = \alpha'_1 - \delta$, gives $g(\vec{\alpha}') > 0$ and $b(\vec{\alpha}')>0$, which precludes $F$ being the negative unweighted Borda rule or weakly Borda with respect to it. Thus $F$ is not weakly Borda with respect to any unweighted Borda rule.

Now by Lemma \ref{weightinglemma} we can further say that $F$ is not weakly Borda with respect to any Borda rule. Finally, the tie rule gives a tie where our SWF gives a decisive ranking at the $\vec{\alpha}$ given above, so $F$ is not the tie rule. This exhausts all possibilities, so $F$ is strongly non-Borda.
\end{proof}

This completes the proof.
\end{proof}

We get the finite case as a corollary.

\begin{thm}\label{FUA}
For large enough $n$, there exists a strongly non-Borda SWF on $V$ with $|V| = n$ and $C = \set{c_1,c_2,c_3}$ satisfying U, MIIA, A, N, P and PR.
\end{thm}

\begin{proof}
For $n$ to be determined later, we let $F_n$ on $V$ with $|V|=n$ be given by the weight-based relative SWF $f'_n(\vec{\alpha})$, where for all $t \in D_3$ we have $\alpha_t = |\set{v \in V: \pi_{i,j}(e(v))= t}|$. We set $f'_n(\vec{\alpha}) = f'(\vec{\alpha}/n)$; in other words, this is the infinite SWF given above, acting on the $n$ voters as though they were $n$ equal parts of the infinite electorate $V$, all voting as blocks. All of the required conditions of the SWF (except for IV) are inherited from the infinite version. What is left is to show that $F_n$ is still strongly non-Borda.

In the proof that the infinite $F$ was strongly non-Borda, we found a vector $\vec{\alpha}$ for which $g(\vec{\alpha})$ and $b(\vec{\alpha})$ had opposite signs. By the continuity of both $b$ and $g$ in $\vec{\alpha}$ we have an open ball $U$ on which this is true. Then for large enough $N$ we can find a volume $T$ contained in $U$ of the form

\begin{equation}
T = \vec{\rho} + \set{\vec{\lambda} | \lambda_t \geq 0, \sum_t \lambda_t = 4/N}
\end{equation}

where $\sum_t \rho_t = 1 - 4/N$. Now for any $n\geq N$ we pick $\vec{\mu}$ such that $\rho_t \leq \mu_t \leq \rho_t + 1/N$ and $n \mu_t \in \mathbbm{Z}$ for $t \neq -2$, and $\mu_{-2} = 1 - \sum_{t \neq -2} \mu_t$. Now $\mu_{-2} \geq 1 - \sum_{t \neq -2} \rho_t - 3/N$, so $\mu_{-2} - \rho_{-2} \geq 1/N > 0$. Thus $\vec{\mu} > \vec{\rho}$, so $\vec{\lambda} = \vec{\mu} - \vec{\rho}$ has $\lambda_t \geq 0$. Additionally, $\sum_t \lambda_t = \sum_t \mu_t - \sum_t \rho_t = 4/N$. Hence $\vec{\mu} = \vec{\rho} + \vec{\lambda} \in T$. Finally, for $t \neq -2$ we have $n \mu_t \in \mathbbm{Z}$, and $n \sum_t \mu_t = n$, so $n \mu_{-2} \in \mathbbm{Z}$. Therefore $n\vec{\mu}$ is in the domain of $f'_n$.

By picking $e_+ \in \tau^{-1} \circ \pi_{i,j}^{-1}(n\vec{\mu})$, we have found an election on $n$ voters (for all $n \geq N$) where the positive Borda rule gives a decisive result and $F_n$ gives the opposite result. But a symmetrical argument shows that on a different election $e_-$, $F_n$ gives the opposite result to the negative Borda rule. Finally, $F_n$ is not weakly Borda with respect to any weighted Borda rule by Lemma \ref{weightinglemma}, and it is not the tie rule as $f'_n(n\vec{\mu})$ is decisive. So for large enough $n$, $F_n$ is a finite strongly non-Borda election, as required.

For completeness we follow this through with an explicit construction. Let $X=2$ and let $n=31$. Then for the comparison between $c_1$ and $c_2$ set $\alpha_2 = 0$, $\alpha_1 = 19$, $\alpha_{-1}= 4$ and $\alpha_{-2} = 8$. The difference between the Borda scores of $c_1$ and $c_2$ is $b(\vec{\alpha}) = 19 - 4 - 2 \cdot 8 = -1 <0$, so $c_2$ beats $c_1$ according to the Borda rule. However, the difference between their scores given by $F_{31}$ is $g((0,19/31,4/31,8/31)) = (65-31 \cdot 2)/31^2 > 0$, so $F_{31}$ ranks $c_1$ above $c_2$. Thus $F_{31}$ is not the positive unweighted Borda rule or the tie rule. But for $\alpha_2 = 31$, $\alpha_i = 0$ for $i \neq 2$, we have $g(\vec{\alpha}) > 0$ and $b(\vec{\alpha}) >0$, so $F_{31}$ is not the negative unweighted Borda rule either. Proceeding as in the general case above, we see that $F_{31}$ is strongly non-Borda as required.
\end{proof}

This completes our investigation into SWFs on the unrestricted domain with complete anonymity between voters; strongly non-Borda examples exist in both the finite and infinite cases due to pathological combinations of the sub-elections on the two Condorcet cycle domains.

\section{Condorcet cycle domain and anonymity}\label{CA}

In this section we prove that on a Condorcet cycle domain, the conditions A, N, MIIA and PR do indeed force the SWF to be the positive unweighted Borda or the tie rule. In the finite case, the PR condition is only needed to rule out the negative unweighted Borda rule. In the infinite case, measurability is sufficient to force $F$ to be an unweighted Borda rule, and PR guarantees that it is not negative, but without measurability there exist pathological counterexamples, even satisfying P. On the other hand, condition P is stronger in the finite cases and in the measurable case, ruling out everything except for the positive unweighted Borda rule.

In all of these cases, note that by conditions A, N and MIIA, and since the domain is a ballot domain,  $F$ must correspond to a unique weight-based relative SWF $f': A' \rightarrow \set{W,T,L}$. Now for all $i$, $A_{i, i+1}$ is the set of measurable functions from $V$ onto $\pi_{i,i+1}(\mathfrak{C}) = \set{1,-2}$; and $A'_{i, i+1}$ is the set of two-dimensional non-negative vectors $\vec{\alpha} = (\alpha_1, \alpha_{-2})$ where $\alpha_1 + \alpha_{-2} = 1$ for $|V|$ infinite or $\alpha_1+\alpha_{-2}=n$ for $|V| = n$. On the other hand, $A'_{i+1,i} = \varrho(A'_{i,i+1})$, and $A' = A'_{i,i+1} \sqcup A'_{i+1,i}$. Now since $\pi_{j,i}(F(e)) = -\pi_{i,j}(F(e))$, we have $f(\pi_{j,i}(e)) = -f(\pi_{i,j}(e))$, so the behaviour of $f'$ on $A'_{i,i+1}$ determines its behaviour on $A'_{i+1,i}$ by $f' = -f' \circ \varrho$, and therefore on all of $A'$. Therefore $F$ is uniquely determined by a function $g: I \rightarrow \set{W,T,L}$, where $I = [0,1]$ in the infinite case and $I = \set{0, \dots, n}$ in the case of $|V| = n$, by

\begin{equation}
\begin{aligned}\label{fprimefromg}
f'(\vec{\alpha}) = g(\alpha_{-2}) \quad \vec{\alpha} \in A'_{i,i+1} \\
f'(\vec{\alpha}) = -g(\alpha_2) \quad \vec{\alpha} \in A'_{i+1,i}
\end{aligned}
\end{equation}

In order to establish or disprove condition PRm we use the following lemma:

\begin{lem}\label{ccprm}
For any $V$, the SWF $F$ satisfies PRm if and only if $g$ is a non-increasing function.
\end{lem}

\begin{proof}
Suppose first that $g$ is increasing. Then suppose that $\vec{\alpha}^+$ and $\vec{\alpha}^-$ are two relative elections in $A'_{i,i+1}$, where $\vec{\alpha}^+ > \vec{\alpha}^-$. Then $\alpha^+_{-2} < \alpha^-_{-2}$. But $f'(\vec{\alpha}) = g(\alpha_{-2})$, and $g$ is non-increasing, so $f'(\vec{\alpha}^+) \geq f'(\vec{\alpha}^-)$ as required for condition PRm. On the other hand, for $\vec{\alpha}^+ > \vec{\alpha}^-$ in $A'_{i+1,i}$, we have $\varrho(\vec{\alpha}^-) > \varrho(\vec{\alpha}^+)$ in $A'_{i,i+1}$, so as we have seen $f'(\varrho(\vec{\alpha}^-)) \geq f'(\varrho(\vec{\alpha}^+))$, so $f'(\vec{\alpha}^-) \leq f'(\vec{\alpha}^+)$ as required. All of $A'$ is made up of $A'_{i,i+1}$ and $A'_{i+1,i}$, so this completes the first direction of the lemma.

Now suppose $g$ is not non-increasing. Then there exists $i > j \in I$ such that $g(i) > g(j)$. Setting $n = |V|$ if $V$ is finite and $1$ otherwise, let $\vec{\alpha}^- = (n-i,i)$ and $\vec{\alpha}^+ = (n-j,j)$. Then $f'(\vec{\alpha}^-) = g(i) > g(j) = f'(\vec{\alpha}^+)$ but $\vec{\alpha}_- < \vec{\alpha}_+$ so condition PRm does not hold. Thus PRm holds if and only if $g$ is non-increasing, as required.
\end{proof}

First we deal with the case of $|V|$ finite.

\begin{thm}\label{FCA}
Any SWF on $V$ finite and $C = \set{c_1,c_2,c_3}$ satisfying CC, MIIA, A and N is an unweighted Borda rule. If P is also met, then the SWF must be the unweighted positive Borda rule. Alternatively, if PR is also met, then the SWF must be the unweighted positive Borda rule or the tie rule.
\end{thm}

In order to prove this result and subsequent results across this paper
and forthcoming papers, 
we will use the following critical lemma.

\begin{lem}\label{numberlineelection}
For $\ell \geq 0$ let $I = \set{0, 1, \dots, \ell}$ and suppose $g: I \rightarrow \set{W,T,L}$. Moreover, let $m \in \mathbbm{Z}$ be such that $\ell \leq m \leq 2\ell$, and suppose that for all $i, j, k$ with $i+j+k = m$, the multiset $\set{g(i), g(j), g(k)}$ is consistent. Then $g$ is given by $g(i) = \varphi(\kappa(i-m/3))$ for some $\kappa \in \mathbbm{R}$.
\end{lem}

\begin{proof}

Suppose first that $m=3n$ for some integer $n$. Now $\ell/3 \leq n \leq 2\ell/3$ so $n \in I$. The multiset $\set{g(n),g(n),g(n)}$ is consistent, but it has only one distinct element, so it must be $g(n)=T$.

On the other hand, for any $m$ we let $n_- = a_0$ be the largest integer less than $m/3$, and we let $n_+$ be the smallest integer greater than $m/3$. We let $g(n_-)=X_-$. Now either $3 \mid m$, in which case $n_- + 1 = m/3 = n_+ - 1$, or $n_+ - m/3 = 1/3$ or $m/3 - n_- = 1/3$. In the first case, we set $b'_0 = m/3$. In the second case, we note that by taking the triple $\set{g(n_-), g(n_+), g(n_+)}$ we have $g(n_+) = -X_-$; then we set $b'_0 = n_-$. In the third case we also set $b'_0 = n_-$.

We observe that if $g(i) = X$ for $a \leq i \leq b$ with $a, b \in \mathbbm{Z}$, and $g(i') \in \set{X, T}$ for $a' \leq i' \leq b'$ with $a',b' \in \mathbbm{Z}$, then $g(j) = -X$ for $m-b-b' \leq j \leq m-a-a'$. Indeed, given $j$, we take $i$ and $i'$ such that $i+i' = m-j$. This is possible as $i+i'$ can range in value from $a+a'$ to $b+b'$. Then $g(i)= X$ and $g(i')\in \set{X,T}$, and $i+i'+j=m$ so $\set{g(i),g(i'),g(j)}$ is consistent; this forces $g(j) = -X$.

Now we carry out the following induction. Suppose that for $a_i \leq t \leq n_- = a_0$ we have $g(t)= X_-$ and for $a_i \leq t \leq b'_0$ we have $g(t) \in \set{X,T}$. (This holds for $i=0$ as a base case). Then applying the observation above, we find that $g(t)=-X_-$ for $m-n_- - b'_0 \leq t \leq m- 2a_i$. In the case where $n_+ - m/3 = 1/3$, we note that $m - n_- - b'_0 = m-2n_- = n_+ +1$, but we know that $g(n_+)=-X_-$, so $g(t)=-X_-$ for $n_+ \leq t \leq m-2a_i$. In the other two cases we have $m - n_- - b'_0 = n_+$, so in all cases we can write $g(t)=-X_-$ for $n_+ \leq t \leq m-2a_i$.

Now we reapply the observation to find that for $4a_i - m \leq t \leq m - 2n_+$ we have $g(t) = X_-$. Now $m-2n_+ \geq n_- -1$ in all three cases, so we can extend this to say that $g(t)=X_-$ for $4a_i - m \leq t \leq n_-$. Setting $a_{i+1} = 4a_i - m$, we know that $a_i \leq a_0 < m/3$, so $3a_i < m$ and $a_{i+1} = 4a_i-m < a_i$. Therefore each step of the induction extends downwards the range of integers up to $n_-$ on which $g$ takes the value $X_-$. This terminates when $a_j <0$, in which case we find that for $0 \leq t \leq n_-$, $g(t)= X_-$; then applying the observation one more time we find that for $n_+ \leq t \leq m$ we have $g(t) = -X_-$.

Therefore $g$ is constant on all $i$ with $i < m/3$, and on all $j$ with $j > m/3$, and $g(n_-) = - g(n_+)$. Finally $g(m/3)=T$ if $3 \mid m$, so $f$ is given by $\varphi$ as described.
\end{proof}

We now proceed to the proof of Theorem \ref{FCA}.

\begin{proof}[Proof of Theorem \ref{FCA}]

Suppose that $F$ is a SWF on $V$ with $|V|=n$ satisfying CC, MIIA, A and N. $F$ corresponds to a function $g: I \rightarrow \set{W,T,L}$ where $I = \set{0, 1, \dots, n}$, such that $f'$ is determined from $g$ by equation \ref{fprimefromg}, and for all $c_i$ and $c_j$ we have $f' \circ \pi_{i,j} \circ \tau = \pi_{i,j} \circ F$.

Now let $\ell = n$ and $m = n$. For $i, j, k$ with $i+j+k = m$, we can define $\vec{\varepsilon} \in E'$ by $\varepsilon_{c_1 c_2 c_3} = i$, $\varepsilon_{c_2 c_3 c_1} = j$ and $\varepsilon_{c_3 c_1 c_2} = k$. Then

\begin{equation}
\begin{aligned}
\pi_{1,2}(F(e)) &= f(\pi_{1,2}(e)) \\
&= f'(\pi_{1,2}(\vec{\varepsilon})) \\
&= f'(i+k,j) = g(j)
\end{aligned}
\end{equation}

Similarly $\pi_{2,3}(F(e)) = g(k)$ and $\pi_{3,1}(F(e))=g(i)$. In order for $F$ to be a weak ordering, we require the multiset $\set{\pi_{1,2}(F(e)), \pi_{2,3}(F(e)), \pi_{3,1}(F(e))} = \set{g(i),g(j),g(k)}$ to be consistent. Therefore $g$ must satisfy the conditions of Lemma \ref{numberlineelection} and must be of the form $g(i) = \varphi(\kappa(i-n/3))$. This corresponds to $f'$ defined by $f'(\vec{\alpha}) = \varphi(-\kappa(\alpha_1 - 2 \alpha_2))$, and in turn the Borda rule $B_w$ with constant weight $w \equiv -\kappa$. Since $\pi_{i,i+1}(F(e))$ and $\pi_{i,i+1}(B_w(e))$ agree for all $i$ and $e$, and since $\pi_{i+1,i}$ is determined for both SWFs by $\pi_{i,i+1}$, we have that $B_w$ and $F$ agree for all $e$, so they are the same SWF. Therefore $F$ is an unweighted Borda rule, concluding the first part of the theorem.

If P is met, then we know that $g(n) = L$, so $\kappa <0$, so $w >0$ and the SWF is the positive Borda rule. However, both the positive rule and the tie rule satisfy PR. What is left is to demonstrate that the negative rule does not satisfy PR. Since $\mathfrak{C}$ is an increasing ballot, we check PRm instead.

Consider the vectors $\vec{\alpha}_- = (0,n)$ and $\vec{\alpha}_+ = (n,0)$ in $A'_{i,i+1}$. We have $\vec{\alpha}_+ > \vec{\alpha}_-$, but if $F$ is the negative Borda rule then for some $\kappa>0$ we have $f'(\vec{\alpha}_-) = \varphi(2\kappa n/3) = W$, and $f'(\vec{\alpha}_+) = \varphi(-\kappa n/3)=L$. Therefore conditon PRm doesn't hold, and so neither does PR. Therefore adding the PR condition rules out the negative case and leaves only the positive Borda rule or the zero-weighted Borda rule under which all results are three way ties.
\end{proof}

We use this result to establish the infinite version of the theorem under measurability assumptions.

\begin{defn}[Measurability for anonymous SWFs]\label{def:measurable}
An anonymous SWF $F = F' \circ \tau$ is said to be measurable if $F'$ is a measurable function from $\set{\vec{\varepsilon} \in \mathbbm{R}_{\geq 0}^{\mathfrak{R}} | \sum_r \varepsilon_r = 1}$, equipped with the Lebesgue measure, onto $\mathfrak{R}_=$ equipped with the discrete measure.
\end{defn}

\begin{thm}\label{ICAmeasurable}
Any measurable SWF on $V$ infinite and $C = \set{c_1,c_2,c_3}$ satisfying CC, MIIA, A and N is an unweighted Borda rule.
\end{thm}

We will use the following Lemma, which will also be used in subsequent papers:

\begin{lem}\label{measurable}
Let $-a < 0 < b$ with $2a \geq b$ and $2b \geq a$, and let $f$ be a measurable function from $[-a,b]$ to $\set{W,T,L}$. Suppose that for any $x_1,x_2,x_3 \in [-a,b]$ with $x_1+x_2+x_3=0$, the multiset $\set{f(x_i)}_i$ is consistent. Then $f(x) = \varphi(\kappa x)$ for some real number $\kappa$.
\end{lem}

\begin{proof}

First note that setting $x_i=0$ for all $i$ gives that the multiset $\set{f(0),f(0),f(0)}$ is consistent, so $f(0)=T$.

Next we find the following discrete copies of the result sitting within $[-a,b]$:

\begin{lem}
For any $\delta > 0$ there is some $X \in \set{W,T,L}$ such that for all $n \in \mathbbm{Z}$ with $0 < n < \min(a,b)/\delta$, we have $f(n\delta) = X$ while $f(-n\delta) = -X$.
\end{lem}

\begin{proof}
Let $c = \min(a,b)$. Let $N = \lfloor  c/\delta \rfloor$, let $\ell = 2N$, $m=3N$ and let $I = \set{0,1,\dots,\ell}$. Now define $g: I \rightarrow \set{W,T,L}$ by $g(k) = f(\delta(k-N))$. Since $k-N$ ranges from $-N$ to $N$, it is always between $\pm c/\delta$, so $\delta(k-N)$ is between $-a$ and $b$ and $f(\delta(k-N))$ is defined. But for $i, j, k$ with $i+j+k=m$, we have $g(i)+g(j)+g(k) = f(\delta(i-N)) + f(\delta(j-N)) + f(\delta(k-N))$. Then since $\delta(i-N) + \delta(j-N) + \delta(k-N) = 0$, we are given that $\set{g(i),g(j),g(k)}$ is consistent. Applying Lemma \ref{numberlineelection} gives that $g(i) = \varphi(\kappa(i-N))$ for some $\kappa$, so $f(n \delta) = \varphi(\kappa n)$, proving the Lemma.
\end{proof}

Consequently we find that for any $x, y \in [-c,c]$ with $x/y \in \mathbbm{Q}_{>0}$, $f(x)=f(y)$; indeed, $x/y = p/q$ with $p, q \in \mathbbm{Z}_{>0}$ and take $\delta = y/q = x/p$. Then $p$ and $q$ are both less than $c/\delta$, so $f(p\delta) = f(q\delta)$. In other words, restricting $f$ to $[-c,c]$, $f^{-1}(W)$, $f^{-1}(T)$ and $f^{-1}(L)$ are invariant under positive rational dilations around $0$. Moreover, by setting $\delta = x$, we find that $f(x)=-f(-x)$.

Now since $f$ is measurable we can consider $\mu_X = \mu(f^{-1}(X) \cap [0,c])$ for each $X \in W, T, L$; then $\mu_W+\mu_T+\mu_L = c$ so we can take some $X$ for which $\mu_X >0$. Therefore there exists some interval $(s,t) \subset [0,c]$ on which the density of $f^{-1}(X)$ is more than $3/4$ - that is, $\mu((s,t) \cap X) > 3(t-s)/4$.

Now for any $x \in [0,c]$, we can find positive rational numbers $r$ arbitrarily close to $x/(s+t)$. We restrict ourselves to $r$ sufficiently close to guarantee $rs, rt \leq c$. Now $f^{-1}(X)$ has density more than $3/4$ on $[rs,rt]$, and we choose $r$ to be close enough to $x/(s+t)$ to ensure that it has density more than $1/2$ on $[rs,x-rs]$ as well (since $x-rs$ tends to $rt$ as $r$ approaches $x/(s+t)$).

Now $f^{-1}(X) \cap (x - f^{-1}(X))$ is the intersection of two sets of density more than $1/2$ on $[rs,x-rs]$, so it is non-empty and contains some value $y$. We have $f(y)=X$ and $f(x-y)=X$. Then the triple $y, x-y, -x$ sums to zero and falls within $[-a,b]$, so by assumption we have $\set{f(y), f(x-y),f(-x)}$ consistent, so $f(-x)=-X$ and $f(x)=X$. But this was true for all $x \in [0,c]$. Similarly $f(x)=-X$ for all $x \in [-c,0]$.

Finally, for any $y \in [-a,b]$, the triple $y, -y/2, -y/2$ sums to zero and falls within $[-a,b]$ (by the inequalities on $a$ and $b$), so $\set{f(y),f(-y/2),f(-y/2)}$ is consistent, and $f(y) = -f(-y/2)$. But $-y/2 \in [-c,c]$ so $f(-y/2)$ is $X$ if $y<0$ and $-X$ if $y>0$; then $f(y)=-X$ if $y<0$ and $X$ if $y>0$, completing the Lemma.
\end{proof}

We are now ready to proceed to the proof of the measurable case.

\begin{proof}[Proof of Theorem \ref{ICAmeasurable}]
As in the finite case, suppose that $F$ is a SWF on $V$ infinite satisfying CC, MIIA, A and N. $F$ corresponds to a function $g: I \rightarrow \set{W,T,L}$ where $I = [0,1]$, such that $f'$ is determined from $g$ by equation \ref{fprimefromg}, and for all $c_i$ and $c_j$ we have $f' \circ \pi_{i,j} \circ \tau = \pi_{i,j} \circ F$. $F$ is  measurable if and only if $g$ is measurable. 

Also as in the finite case, for any $i, j, k \in [0,1]$ with $i+j+k=1$ we can set $\vec{\varepsilon} \in E'$ by $\varepsilon_{c_1 c_2 c_3} = i$, $\varepsilon_{c_2 c_3 c_1} = j$ and $\varepsilon_{c_3 c_1 c_2} = k$. Then $\pi_{1,2}(F(e)) = g(j)$, $\pi_{2,3}(F(e)) = g(k)$ and $\pi_{3,1}(F(e))=g(i)$. In order for $F$ to be a weak ordering, we require the multiset $\set{\pi_{1,2}(F(e)), \pi_{2,3}(F(e)), \pi_{3,1}(F(e))}$ to be consistent. This means that $\set{g(i),g(j),g(k)}$ must be consistent.

We define $h(i)=g(i+1/3)$ for $i \in [-1/3,2/3]$. Then the conditions for Lemma \ref{measurable} apply to $h$, so $h(x) = \varphi(\kappa x)$, or $g(i) = \varphi(\kappa (x-1/3))$, and $f'(\vec{\alpha}) = \varphi((-\kappa/3)\cdot(\alpha_1 - 2\alpha_{-2}))$, so $F$ is an unweighted Borda rule, as required.
\end{proof}

This immediately gives Maskin's result:

\begin{cor}[Maskin]\label{ICA}
Any SWF on $V$ and $C = \set{c_1,c_2,c_3}$ satisfying CC, MIIA, A, N and PR is the positive unweighted Borda rule or the tie rule.
\end{cor}

\begin{proof}
Lemma \ref{ccprm} means that condition PR guarantees that $g$ is non-increasing. Hence for all $g(\ell) = L, g(t) = T, g(w) = W$ we have $\ell < t < w$. Clearly $g^{-1}(L)$ is either empty or an interval, and the same goes for $g^{-1}(T)$ and $g^{-1}(W)$; so $g$ is measurable. Thus Theorem \ref{ICAmeasurable} guarantees that $F$ is an unweighted Borda rule. Since $g$ is non-increasing, we find that the rule must be the positive unweighted Borda rule or the tie rule, as in the proof of Theorem \ref{FCA}.
\end{proof}

Finally, without the measurability condition (or the stronger PR condition) we construct pathological counterexamples, including examples satisfying the Pareto condition.

\begin{thm}\label{ICAchoice}
Assuming the axiom of choice, there exists a strongly non-Borda SWF on $V$ infinite satisfying CC, MIIA, A, N and P.
\end{thm}

The logic of the construction is to find a partition of $I$ into wins, ties and losses such that consistency is satisfied, but using unmeasurable sets to avoid the Borda rule. The proof of Lemma \ref{measurable} demonstrated that these sets would be closed under addition and rational dilatons, so we use Zorn's lemma to construct such sets.

\begin{proof}

Consider the family of sets $W' \subset \mathbb{R}^*$ satisfying the following conditions:
\begin{enumerate}
\item\label{nonborda} $1, -\sqrt{2} \in W'$
\item\label{rationaldilation} If $t \in W'$ then $qt \in W'$ for all positive rational numbers $q$
\item\label{closed} If $s, t \in W'$ then $s+t \in W'$
\end{enumerate}

This family forms a partially ordered set $P$ under the order relation of set inclusion. Note that $P$ is non-empty because $\set{q_1-q_2\sqrt{2} | q_1, q_2 \in \mathbb{Q}_{\geq 0}, q_1 + q_2 > 0}$ is in $P$. Condition \ref{nonborda} is satisfied by $(q_1,q_2) = (1,0)$ and $(q_1,q_2)=(0,1)$. Conditions \ref{rationaldilation} and \ref{closed} are satisfied trivially.

Any non-empty chain $C \subset P$ has an upper bound which is the union $U$ over all of $C$. This is clearly a subset of $\mathbb{R}^*$, and it contains $1$ and $-\sqrt{2}$ from any of its members so it satisfies condition \ref{nonborda}. Now if $t \in U$ then $t \in C_t$ for some $C_t \in C$, and then $qt \in C_t$ for all positive rational $q$, so $qt \in U$ and $U$ satisfies condition \ref{rationaldilation}. Finally, if $s, t \in U$ then $s \in C_s$ and $t \in C_t$ and without loss of generality $C_s \subset C_t$, so $s,t \in C_t$ and $s+t \in C_t$, so $s+t \in U$, and $U$ satisfies condition \ref{closed}.

Hence by (the non-empty formulation of) Zorn's Lemma the set $P$ contains at least one maximal element, which we call $M$. We will show that $M \cup -M = \mathbb{R}^*$.

Suppose instead that some $s \in \mathbb{R}^*$ has $s \notin M \cup -M$. Now we construct

\begin{equation}
M' = M \cup \bigcup_{q \in \mathbb{Q}_{>0}} qs+(M \cup \set{0})
\end{equation}

Firstly, $0 \notin M'$ because $0 \notin M$ and if $0 \in qs + (M \cup \set{0})$ for some $q$, then $-qs \in M \cup \set{0}$. Then $-qs \in M$ and $-s \in M$, so $s \in -M$ contradicting our definition of $s$. Hence $M' \subset \mathbb{R}^*$.

Now clearly condition \ref{nonborda} is inherited from $M$.

Condition \ref{rationaldilation} is met because if $t \in M'$ then either $t \in M$ in which case $qt \in M \subset M'$, or $t = rs + m$ for $r$ positive rational and $m \in M \cup \set{0}$; then $qt = qrs + qm$. Now $qm \in M \cup \set{0}$ because either $m \in M$ in which case $qm \in M$ or $m = 0$ in which case $q \cdot 0 = 0 \in M \cup \set{0}$, and $qr$ is a positive rational, so $qt \in M'$.

Finally, condition \ref{closed} holds because if $u,v \in M'$ then either:

\begin{itemize}
\item $u, v \in M$ and so $u+v \in M \subset M'$, or
\item $u \in M$ and $v = qs+m$ for $m \in M \cup \set{0}$, in which case $u+m \in M$ and $u+v = qs + (u+v) \in M'$, or the same with $u$ and $v$ reversed, or
\item $u = q_1 s+m_1$ and $v=q_2 s+m_2$, so $u+v = (q_1+q_2)s + (m_1+m_2)$, and $m_1+m_2 \in M \cup \set{0}$, so $u+v \in M'$
\end{itemize}

So $M'$ satisfies conditions \ref{nonborda}, \ref{rationaldilation} and \ref{closed}.

Thus $M' \in P$, but $M'$ is greater than $M$ in the partial order of set inclusion, contradicting the fact that $M$ is maximal; so in fact there was no such $s$, and it was the case that $M \cup -M = \mathbbm{R}^*$.

Now we define $g: [0,1] \rightarrow \set{W,T,L}$. We set $g(1/3)=T$, for $x-1/3 \in M$ we set $g(x)=L$, and otherwise we set $g(x)=W$. Defining $f'$ from $g$ via equation \ref{fprimefromg}, we then get a SWF $F$ if $f'$ meets the conditions set out in Lemma \ref{relativetofull}.

Condition \ref{relativetofull:twoway} in the Lemma is met by definition, because the value of $f$ on $a \in A_{i+1,i}$ is defined by $f(-a)=-f(a)$. So what is left is to check condition \ref{relativetofull:threeway}.

Consider the election vector $\vec{\varepsilon} = (\varepsilon_{c_1 c_2 c_3},\varepsilon_{c_2 c_3 c_1},\varepsilon_{c_3 c_1 c_2})$, which we will call $(\varepsilon_1,\varepsilon_2,\varepsilon_3)$ for brevity. Then 

\begin{equation}
\begin{aligned}
\pi_{1,2}(\vec{\varepsilon}) &= (\varepsilon_1+\varepsilon_3, \varepsilon_2) \\
\pi_{2,3}(\vec{\varepsilon}) &= (\varepsilon_1+\varepsilon_2, \varepsilon_3) \\
\pi_{1,2}(\vec{\varepsilon}) &= (\varepsilon_2+\varepsilon_3, \varepsilon_1)
\end{aligned}
\end{equation}

Thus $f'(\pi_{1,2}(\vec{\varepsilon})) = g(\varepsilon_2)$, and so on for other indices, so it remains to show that for $\varepsilon_1+\varepsilon_2+\varepsilon_3=1$ all non-negative, we have that $\set{g(\varepsilon_i)}_i$ is consistent.

First, if any $\varepsilon_i = 1/3$, then $g(\varepsilon_i)=T$. Now there are three cases:
\begin{enumerate}
\item If $\varepsilon_{i+1}=1/3$ then $\varepsilon_{i-1}=1/3$ and the multiset is $\set{T,T,T}$, which is consistent.
\item If $\varepsilon_{i+1}\neq 1/3$ and $\varepsilon_{i+1} -1/3 \in M$ then $g(\varepsilon_{i+1}) = L$. Meanwhile $\varepsilon_{i-1} \neq 1/3$, and if  $\varepsilon_{i-1}-1/3 \in M$ then by additivity we have $(\varepsilon_{i+1}-1/3)+(\varepsilon_{i-1}-1/3)=0 \in M$, a contradiction. Thus $g(\varepsilon_{i-1}) = W$. Thus the multiset is $\set{W,T,L}$, which is consistent.
\item If $\varepsilon_{i+1}\neq 1/3$ and $\varepsilon_{i+1} -1/3 \notin M$ then $g(\varepsilon_{i+1})=W$. Now since $M \cup -M = \mathbbm{R}^*$, we have $\varepsilon_{i-1}-1/3 = -(\varepsilon_{i+1}-1/3) \in M$. Then $g(\varepsilon_{i-1}) = L$. Thus the multiset is $\set{W,T,L}$, which is consistent.
\end{enumerate}

Suppose instead that $\varepsilon_i \neq 1/3$ for all $i$. Then $g(\varepsilon_i) = L$ or $g(\varepsilon_i) = W$ for all $i$. The only two inconsistent possibilities for $\set{g(\varepsilon_i)}_i$ are $\set{W,W,W}$ and $\set{L,L,L}$.

Suppose that the multiset is $\set{L,L,L}$. Then $\varepsilon_i-1/3 \in M$ for all $i$. But by additivity, this means that $\sum_i (\varepsilon_i -1/3) = 0 \in M$, a contradiction.

Now suppose that the multiset is $\set{W,W,W}$. Then $\varepsilon_i-1/3 \notin M$ so $1/3 - \varepsilon_i \in M$ for all $i$; so again by additivity we find that $0 \in M$, a contradiction.

So instead we find that $\set{g(\varepsilon_i)}_i$ is consistent, so $F$ is a valid SWF, defined on the Condorcet Cycle domain, and satisfying A, N and MIIA as required. Finally we need to check that $F$ satisfies condition P, and that it is strongly non-Borda.

Suppose that $c_i$ is ranked above $c_{i+1}$ by all voters. Then $\pi_{i,i+1}(e(v))>0$. But $\pi_{i,i+1}(\mathfrak{R}) = \set{1,-2}$, so this requires  $\pi_{i,i+1}(e) \equiv 1$, or $\vec{\alpha}$ defined by $\alpha_1 = 1$. Then $f'(\vec{\alpha}) = g(\alpha_{-2})=g(0)$. But since $1 \in M$, if we had $-1/3 \in M$ then by additivity we would have $0 \in M$, a contradiction; so $-1/3 \notin M$ and we get $g(0)=W$. Thus $f'(\vec{\alpha})=W$, as required by condition P.

On the other hand, if $c_{i+1}$ is ranked above $c_i$ by all voters, then $\pi_{i,i+1}(e(v))<0$ for all $v\in V$. This is corresponds to a $\vec{\alpha}$ defined by $\alpha_{-2}=1$. Then $f'(\vec{\alpha})=g(\alpha_{-2})=g(1)$. Now $1 \in M$ so $2/3 \in M$ by condition \ref{rationaldilation}, and so $g(1) = L$. Thus $f'(\vec{\alpha})=L$, and by the definition of $f'$ on $A'_{i+1,i}$ we get that $f'(\pi_{i+1,i}(\vec{\varepsilon})) = W$, as required. Therefore the Pareto condition holds.

Finally, note that $F$ is strongly non-Borda. We have seen that $g(1)=L$ and we also know that $g(\sqrt{2}/4 + 1/3) = W$, because $\sqrt{2}/4 \notin M$ (or else, by additivity, $-\sqrt{2} + 4 \cdot \sqrt{2}/4 = 0 \in M$). Correspondingly, we find that on $A'_{i,i+1}$, we have $f'((0,1)) = L$ and $f'((2/3 - \sqrt{2}/4, \sqrt{2}/4 + 1/3)) = W$. But for an unweighted Borda rule with weight $w$, the difference between the Borda scores of candidates $c_i$ and $c_{i+1}$ in these cases are $-2w$ and $-3w\sqrt{2}/4$ respectively; and since these have the same sign, $\pi_{i,i+1} \circ B_w$ must give the same ranking in both cases. The only way that this could happen while $F$ agrees with $B_w$ except as a tiebreaker would be if $w=0$; but by definition we do not call a SWF weakly Borda with respect to the tie rule. Finally, $F$ is not the tie rule since any $t \in [0,1]$ with $t \neq 1/3$ has $g(t) \neq T$; so for $\vec{\alpha} \in A'_{i,i+1}$ other than $(2/3,1/3)$ we have $f'(\vec{\alpha})\neq T$.

Now by Lemma \ref{weightinglemma}, since $F$ must obey condition A, if $F$ were weakly Borda with respect to some weighted Borda rule $B_w$, it would also be weakly Borda with respect to an unweighted Borda rule. This has been ruled out, so in fact $F$ is strongly non-Borda.
\end{proof}

Note that we cannot hope to find such an $F$ without resorting to choice (or some axiom stronger than dependent choice) because the Solovay model \autocite{solovay} of set theory has dependent choice, but every set of the real numbers is measurable. On the other hand, Tao \autocite{tao} has constructed non-measurable sets using the existence of non-principal ultrafilters on $\mathbb{N}$, which is not implied by the axiom of dependent choice but is weaker than the full axiom of choice. Whether these non-measurable sets could be controlled to give closure under addition, as we require here, is an interesting question beyond the scope of this paper.

\section{Acknowledgements}

I am grateful to Eric Maskin for a helpful conversation, and for his valuable comments on the proof of Theorem \ref{IUA}.

\printbibliography

\end{document}